\RequirePackage{fix-cm}
\documentclass[smallextended]{svjour3}       
\smartqed  
\usepackage{graphicx}
\usepackage{upgreek}
\usepackage[british]{babel}
\usepackage{amsmath,amssymb}
\usepackage{epstopdf}
\usepackage[compatibility=false,font=small,labelfont=bf]{caption} 
\usepackage{subcaption}
\usepackage[usenames,dvipsnames]{color}
\usepackage{xcolor}
\usepackage[normalem]{ulem}
\usepackage{hyphenat}
\usepackage{natbib}
\usepackage{cite}
\usepackage{url}
\usepackage{rotating}
\usepackage{xspace}
\usepackage{booktabs}

\makeatletter
\providecommand*{\cupdot}{%
  \mathbin{%
    \mathpalette\@cupdot{}%
  }%
}
\newcommand*{\@cupdot}[2]{%
  \ooalign{%
    $\m@th#1\cup$\cr
    \sbox0{$#1\cup$}%
    \dimen@=\ht0 %
    \sbox0{$\m@th#1\cdot$}%
    \advance\dimen@ by -\ht0 %
    \dimen@=.5\dimen@
    \hidewidth\raise\dimen@\box0\hidewidth
  }%
}

\providecommand*{\bigcupdot}{%
  \mathop{%
    \vphantom{\bigcup}%
    \mathpalette\@bigcupdot{}%
  }%
}
\newcommand*{\@bigcupdot}[2]{%
  \ooalign{%
    $\m@th#1\bigcup$\cr
    \sbox0{$#1\bigcup$}%
    \dimen@=\ht0 %
    \advance\dimen@ by -\dp0 %
    \sbox0{\scalebox{2}{$\m@th#1\cdot$}}%
    \advance\dimen@ by -\ht0 %
    \dimen@=.5\dimen@
    \hidewidth\raise\dimen@\box0\hidewidth
  }%
}
\makeatother

\DeclareMathOperator{\diag}{diag}
\newcommand{\normal}{\mathfrak{n}} 
\newcommand{\dd}{\partial}
\newcommand{\dif}{\mathop{}\!\mathrm{d}} 
\newcommand{\vect}[1]{\mathbf{#1}}
\newcommand{\mat}[1]{\mathbf{#1}}

\newcommand {\nor} [1]{\left\|#1\right\|}

\newcommand{\D}{\mathcal{D}}
\newcommand{\I}{\mathcal{I}}
\newcommand{\R}{\mathbb{R}}
\newcommand{\N}{\mathbb{N}}
\newcommand{\T}{\mathsf{T}} 
\renewcommand{\P}{\mathcal{P}} 

\newcommand{\Ard}{\mathcal{A}} 
\newcommand{\borel}{\mathfrak{B}} 
\newcommand{\paramSet}{\mathcal{S}} 

\def\PAI1{\mbox{PAI-1}\xspace}
\def\istate{\mbox{$i$-state}\xspace}

\newtheorem{convention}{Convention}
\journalname{Journal of Mathematical Biology}
\begin{document}

\title{Structured Models of Cell Migration Incorporating Molecular Binding Processes
}
%


\author{Pia~Domschke \and 
            Dumitru~Trucu \and 
            Alf~Gerisch \and 
            Mark~A.\,J.~Chaplain 
}


\institute{Pia Domschke \and Alf Gerisch \at
              Fachbereich Mathematik, Technische Universit\"at Darmstadt \\
              Dolivostr. 15, 64293 Darmstadt, Germany\\
              \email{domschke@mathematik.tu-darmstadt.de}           \\
              \email{gerisch@mathematik.tu-darmstadt.de}\\
           \and
           Dumitru Trucu \at
              Division of Mathematics, University of Dundee\\
              Dundee DD1 4HN, United Kingdom\\
              \email{trucu@maths.dundee.ac.uk}\\
           \and
           Mark A.\,J. Chaplain \at
	      School of Mathematics and Statistics, Mathematical Institute, University of St Andrews\\
	      St Andrews KY16 9SS, United Kingdom\\
              \email{majc@st-andrews.ac.uk}
}

\date{Received: date / Accepted: date}

\maketitle

\begin{abstract}
The dynamic interplay between collective cell movement and the various molecules involved in the accompanying cell signalling mechanisms plays a crucial role in many biological processes including normal tissue development and pathological scenarios such as wound healing and cancer. Information about the various structures embedded within these processes allows a detailed exploration of the binding of molecular species to cell-surface receptors within the evolving cell population. In this paper we establish a general \emph{spatio-temporal-structural} framework that enables the description of molecular binding to cell membranes coupled with the cell population dynamics. We first provide a general theoretical description for this approach and then illustrate it with two examples arising from cancer invasion. 

\end{abstract}

\section{Introduction}
\label{intro}
 
The modelling of complex biological systems has witnessed extensive developments over the past four decades. Ranging from studying large-scale collective behaviour of inter-linked species in ecological studies to the understanding of the complicated multiscale processes arising in animal and human cell and tissue biology, the modelling has gradually evolved in scope and focus to include not only temporal and spatial coordinates but also structural information of the individual species involved, such as age, size or other relevant quantifiable aspects \citep{Diekmann1982,Metz1986}. 

Spatio-temporal models, in particular reaction-diffusion-taxis systems, have a long history not only in mathematical biology research \citep{Skellam1951} but in the wider applied mathematics community. Such modelling approaches have generally avoided incorporating any structural information in them (e.g. age, size). The development of structured-population models also has a long tradition going back to the seminal work of \citet{Foerster1959}. 
Areas of interest for structured-population modelling include ecology, epidemiology, collective cell movement in normal tissue dynamics and pathological situations, such as malignant solid tumours and leukaemia, to name a few. 
A majority of these models have been concerned with coupling time and structure (e.g. age, size) in individual or collective species dynamics \citep{Trucco1965a,Trucco1965b,Sinko1967,Gyllenberg1982,Diekmann1984,Kunisch1985,Gyllenberg1986,Gyllenberg1987,Tucker1988, Diekmann1992,Diekmann1994,Huyer1994, Calsina1995, Roos1997, Cushing1998,Basse2007,Chapman2007}. Models coupling space and structure were also developed \citep{Gurtin1981,MacCamy1981,Diekmann1982, Garroni1982, Huang1994, Rhandi1998,Langlais2003,Ayati2006a,Delgado2006,Allen2009}, and these have paved the way towards modelling approaches that couple time, space, and structure, opening a new era in the modelling of biological processes \citep{Di-Blasio1979, Busenberg1983,Langlais1988, Fitzgibbon1995, Rhandi1999,So2001,Al-Omari2002,Cusulin2005,Deng2006}. 

Central to the study of structured population models, is the role played by the semigroup framework \citep{Webb1985,Metz1986,Gyllenberg1990,Diekmann1992}.  
Approaches based on delay-differential equations explore the behaviour of the system under consideration in the presence of age, size, or various other appropriate structural information \citep{Mackey1977, Angulo2012}. Questions regarding the  spatio-structural controllability in single species population models have also been addressed by \citet{Gyllenberg1983,Ainseba2000, Ainseba2001,Gyllenberg2002}. Discrete spatial or temporal and continuous in structure models have been equally employed to understand various ecological processes \citep{Gyllenberg1997,Gyllenberg1997a,Matter2002,Camino-Beck2009, Lewis2010}.
These methodologies have been recently complemented with novel measure theory approaches such as the ones proposed by \citet{Gwiazda2010}.
Finally, numerical explorations and computational simulations have also become increasingly present within the range of methods for the analysis of temporal-structural, spatio-structural, and \emph{spatio-temporal-structural} models \citep{Ayati2000, Ayati2006, Ayati2002,Abia2009}.     

Of recent interest is the exploration of structural information within the context of modelling the complex links between cell movement and the cascade of signalling pathway mechanisms appearing within diseases like cancer, both in malignant solid tumours \citep{Basse2003,Basse2004,Basse2005,Basse2007,Ayati2006,Daukste2012, Gabriel2012} and leukaemia \citep{Bernard2003,Foley2009,Roeder2009}, as well as in hematopoetic diseases such as autoimmune hemolytic anemia \citep{Belair1995,Mahaffy1998}.

Although much progress has been made through \emph{in vivo} and \emph{in vitro} research, understanding more deeply the cross-talk between signalling molecules and the individual and collective cell dynamics in human tissue remains a major challenge for the scientific community. The development of a suitable theoretical framework coupling dynamics at the cell population level with dynamics at the level of cell-surface receptors and molecules is crucial in understanding many important normal and pathological cellular processes. To this end, despite all the experimental advancements, mathematical modelling coupling cell-scale structural information with spatial and temporal dynamics is still in its very early days, with only a few recent works on the subject such as those proposing an age-structured spatio-temporal haptotaxis modelling in tumour progression \citep{Walker2007,Walker2008,Walker2009} as well as those addressing the link between age structure and cell cycle and proliferation \citep{
Gabriel2012,Billy2014} or exploring the role of membrane inhomogeneities for individual cells' deformation mechanics \citep{Mercker2013}. However, none of these modelling attempts have addressed so far the coupling between the collective cell movement and the structural binding behaviours enabled by the various molecular signalling pathways that may come under consideration in the overall tissue dynamics.

In general, modelling the coupling between the collective cell dynamics and the contribution of the structural parameters of the signalling molecules travelling along with the moving cell population remains a difficult open question. In this work, we address this question by establishing the fundamentals of a general framework that captures the overall coupled interaction of a \emph{spatio-temporal-structural} cell population density accompanied by a number of binding spatio-temporal molecular species concentrations. This explores the binding, activation and inhibition processes between cell surface-bound and free molecular species and their effect on the overall cell-population dynamics.
 
The paper is structured as follows. In Section~\ref{sec:generalModel} we introduce the framework by deriving from the first principle the general structured model. 
From this, we will derive a corresponding non-structured model by integrating over the structure space. 
Section~\ref{sec:numericalExample} is dedicated to a generic model of cancer invasion. We show the influence of the structure on this very simple model and compare it to the corresponding non-structured model via numerical examples. In Section~\ref{sec:uPASystem}, this novel framework is applied to the more involved case of the uPA system, an enzymatic system that plays an important role in cancer invasion. Finally, in Section~\ref{sec:conclusion} we discuss the new framework and give insights for further developments. 


\section{General Spatio-Temporal-Structured Population Framework} \label{sec:generalModel}

In this section we establish a general framework for our spatio-temporal-structured population model that enables the coupling of cell surface-bound reaction processes with the overall cell population dynamics.

Let $\D\subset \R^d$, $d\in\{1,2,3\}$, be a bounded spatial domain,  $\I = (0,T]$, $0<T\in\R$, be the time interval, and $\P\subset\R^p$, $p\in\N$, be a convex domain of admissible structure states that contains $\boldsymbol 0 \in\R^{p}$ as accumulation point. The set~$\P$ will be referred to as the \istate space \citep{Metz1986} (= \emph{individual's state}). 
Here the temporal, spatial, and structural variables are~$t$,~$x$, and~$y$, respectively. 
Our basic model consists of the following dependent variables: 

\begin{itemize}
\item the structured cell density $c(t,x,y)$, with $(t,x,y) \in \I\times\D\times\P$;
\medskip
\item the extracellular matrix (ECM) density $v(t,x)$, with $(t,x) \in \I\times\D$; 
\medskip
\item $q$ free molecular species, of concentration $m_{i}(t,x)$, with $(t,x) \in \I\times\D$, $i=1,\dots,q$, which may be written in vector notation 
\[
\vect m = (m_1,\dots,m_q)^\T:\I\times\D\to \R^q\,.
\]
\end{itemize}
We consider that~$p$ of the free molecular species are able to bind to the surface of the cells; without loss of generality, these are~$m_{i}$, $i=1,\dots,p$, with $p\leq q$. 
Note that the number $p$ of molecular species being able to bind to a cell's surface corresponds to the dimension of the \istate  space $\P$.
Similar to size-structured population models, see for example \citet{Chapman2007,Diekmann1984} or \citet{Tucker1988}, we model the surface concentration of bound molecules on the surface of the cells by the structure or \istate variable $y\in\P$. This gives rise to the structured cell density $c(t,x,y)$, which denotes the cell number density at a time $t$ of cells at a spatial point $x$ that have a surface concentration equal to $y$ of molecules bound to their surface. Hence, the unit of $c(t,x,y)$ is number of cells per unit volume in space (at $x$) per unit volume in the \istate (at $y$). The surface concentrations~$y_i$, $i=1,\dots,p$, are measured in $[\upmu \rm{mol/cm^2}]$, which yields a unit volume in the $p$-dimensional \istate~$y$ of $[(\upmu \rm{mol/cm^2})^p]$ and thus the unit of the structured cell density $c$ is given by $[\rm{cells/(cm^3\cdot (\upmu \rm{mol/cm^2})^p}]$.

The total, that is non-structured, cell density $C$ at~$t$ and~$x$ is then obtained by integrating the structured cell density over all \istate{}s $y\in\P$,
\begin{align}
   C(t,x) &= \int\limits_\P c(t,x,y) \dif y\,, \label{equ:totalAmountOfCells}
\end{align}
and its unit is therefore given in $[\rm{cells/cm^3}]$.

The \emph{structured cell surface density}~$s(t,x,y)$, in contrast to the above structured cell \emph{number} density, gives, per unit volume in space and per unit volume in the \istate, the surface area of the cells at $t$ and $x$ which have surface concentration $y$. Let us assume that all cells have the same fixed cell surface area~$\varepsilon$ with unit $[\rm{cm^2/cell}]$. Then the structured cell surface density can be expressed as
\begin{align}
   s(t,x,y) &= \varepsilon c(t,x,y)\, \label{equ:surfaceDensity}
\end{align}
and has unit $[\rm{cm^{2}/(cm^3\cdot (\upmu mol/cm^2)^p)}]$.

We are also interested in the \emph{bound molecular species volume concentration} at given $t$ and $x$, denoted by $\vect n(t,x)$. 
Multiplication of the structured cell surface density $s(t,x,y)$ with the respective surface concentration $y$ yields the \emph{structured volume concentration} of the bound molecular species per unit volume in the \istate.
Thus, integration of this structured volume concentration over the \istate space $\P$ yields the desired bound molecular species volume concentration, i.e.,
\begin{align}
   \vect n(t,x) &= (n_1(t,x), \dots, n_p(t,x))^\T:= \int\limits_{\P} y s(t,x,y) \dif y  \quad \in \R^p\,. \label{equ:totalConcentrationReactants}
\end{align}
The unit of~$n_i$, $i=1,\dots,p$, is $[\rm{\upmu mol/cm^3}]$ = $[\rm{n M}]$, which is the same as the unit for the \emph{free molecular species volume concentrations} $m_j$, $j=1,\dots,q$.

Finally, by the density of the ECM we refer to the \emph{mass density of the fibrous proteins inside the ECM}, for example collagen, hence the unit of the ECM density is $[\rm{mg/cm^3}]$.

For a compact notation, we define the combined vector of the structured cell density and the ECM density 
as well as the combined vector of bound and free molecular species volume concentrations by
\begin{align}
  \label{equ:def_u_r}
  \vect u(t,x) &:= \begin{pmatrix} c(t,x,\cdot) \\[1mm] v(t,x)\end{pmatrix} : \P\to\R^2
  \quad\text{and}\quad 
  \vect r(t,x)  := \begin{pmatrix} \vect n(t,x) \\[1mm] \vect m(t,x) \end{pmatrix} \in \R^{p+q}\,, 
\end{align}
respectively.

Since some of the processes modelled are limited due to spatial constraints, we define the \emph{volume fraction of occupied space} by
\begin{align}
   \hat \rho(t,x) \equiv \rho(C(t,x),v(t,x)) := \vartheta_c C(t,x) + \vartheta_v v(t,x) \label{equ:rho}
\end{align}
with suitable parameters~$\vartheta_c$ and~$\vartheta_v$.
Note that with this definition we assume the amount of free and bound molecular species to be negligible for the volume fraction of occupied space.
In the following, we discuss the model equations for the evolution of~$c$, $v$ and~$\vect m$.

\begin{remark}
The quantities defined above can be interpreted in a measure-theo\-retic framework and so can terms in the equations presented in the following subsections. 
For example, the bound molecular species volume concentrations $\vect n(t,x)$, see Eq.~\eqref{equ:totalConcentrationReactants}, can be seen as an expected value and the definition of the binding and unbinding rates in the structural flux, see the discussion in the end of Section~\ref{sec:generalCellPopulation} below, becomes more general in such a context.
We refer the interested reader to Appendix~\ref{app:bindingMeasures}, where we elucidate these issues in some detail.
\end{remark}

\subsection{Cell population} \label{sec:generalCellPopulation}

Consider, inside the \emph{spatio-structural} space~$\D\times\P$, an arbitrary control volume $V\times W$, satisfying that $V$ and $W$ are compact with piecewise smooth boundaries $\partial V$ and $\partial W$. The total amount of cells in~$V\times W $ at time~$t$ is 
\[
   C_{V\times W} (t) = \int\limits_W \int\limits_V c(t,x,y) \dif x \dif y\,.
\]
Per unit time, the rate of change in~$C_{V\times W}$ is given by the combined effect of the sources of cells of the structural types considered over the control volume and the flux of cells into the control volume over the spatial and structural boundaries. Therefore, we have the integral form of the balance law given by
\begin{align}
\begin{split}
   \frac{d C_{V\times W}}{dt} =& \underbrace{\int\limits_W\int\limits_V S(t,x,y) \dif x \dif y }_{\text{source}} 
   - \underbrace{\int\limits_W \int\limits_{\partial V} F(t,x,y)\cdot \normal(x) \dif \sigma_{n-1}(x) \dif y }_{\text{flux over spatial boundary}} \\
   &- \underbrace{\int\limits_V \int\limits_{\partial W} G(t,x,y)\cdot \normal(y) \dif \sigma_{p-1}(y) \dif x }_{\text{flux over structural boundary}}\,,
\end{split} \label{equ:generalModelDerivation}
\end{align}
where~$\sigma_{n-1}$ and~$\sigma_{p-1}$ are the surface measures on~$\partial V$ and~$\partial W$, respectively.
Assuming that the vector fields $F$ and $G$ are continuously differentiable and since $V$ and $W$ are compact with piecewise smooth boundaries, the divergence theorem yields
\begin{align}
\begin{split}
 \frac{d C_{V\times W}}{dt} = &\int\limits_W\int\limits_V S(t,x,y) \dif x \dif y 
   - \int\limits_W \int\limits_V \nabla_x \cdot F(t,x,y) \dif x\dif y \\
   &- \int\limits_V \int\limits_W \nabla_y \cdot G(t,x,y) \dif y \dif x\,.
\end{split} \label{equ:generalModelIntegralForm}
\end{align}
Assuming further that~$c$ and~$c_t$ are continuous, 
Leibniz's rule for differentiation under the integral sign \citep{Halmos1978} gives 
\begin{align}
\begin{split}
   \int\limits_W \int\limits_V \frac{\dd}{\dd t} c(t,x,y) \dif x \dif y 
   =& \int\limits_W\int\limits_V S(t,x,y) \dif x \dif y 
   - \!\int\limits_W \int\limits_V \nabla_x \cdot F(t,x,y) \dif x\dif y \\
   &- \int\limits_W \int\limits_V \nabla_y \cdot G(t,x,y) \dif x \dif y\,.
\end{split} \label{equ:generalModelIntegralForm2}
\end{align}
Since this holds for arbitrary control volumes $V\times W$, we obtain the following partial differential equation, i.e. the corresponding differential form of the balance law for the structured cell density:
\begin{align}
   \frac{\dd}{\dd t} c(t,x,y) &= S(t,x,y) - \nabla_x \cdot F(t,x,y) - \nabla_y \cdot G(t,x,y)\,. 
   \label{equ:generalPDECells}
\end{align}
This form is similar to models of velocity-jump processes, where the \istate describes the velocity and potentially other internal states of an individual, see, for example, \citet{Othmer1988,Erban2005,Xue2009,Xue2011,Kelkel2012,Othmer2013,Engwer2015,Xue2015}.

\paragraph{Source.} The source of the cell population is given by the proliferation of the cells through cell division (there may be other cell sources, even negative ones such as apoptosis, but here we only consider cell division). Let~$\Phi(y,\vect u)$ be the rate at which cells undergo mitosis (proliferation rate).
Similar to equal mitosis in size-structured populations as was considered by \citet{Perthame2007}, we assume that, as cells divide, the daughter cells share the~$p$ different molecular species on their surface equally. That means that a cell at~$(t,x,y)$ divides into two cells at~$(t,x,\frac12 y)$, and a schematic of this can be seen in Figure~\ref{fig:mitosis}. 

\begin{figure}[htb]
\centering
\includegraphics[width=\textwidth]{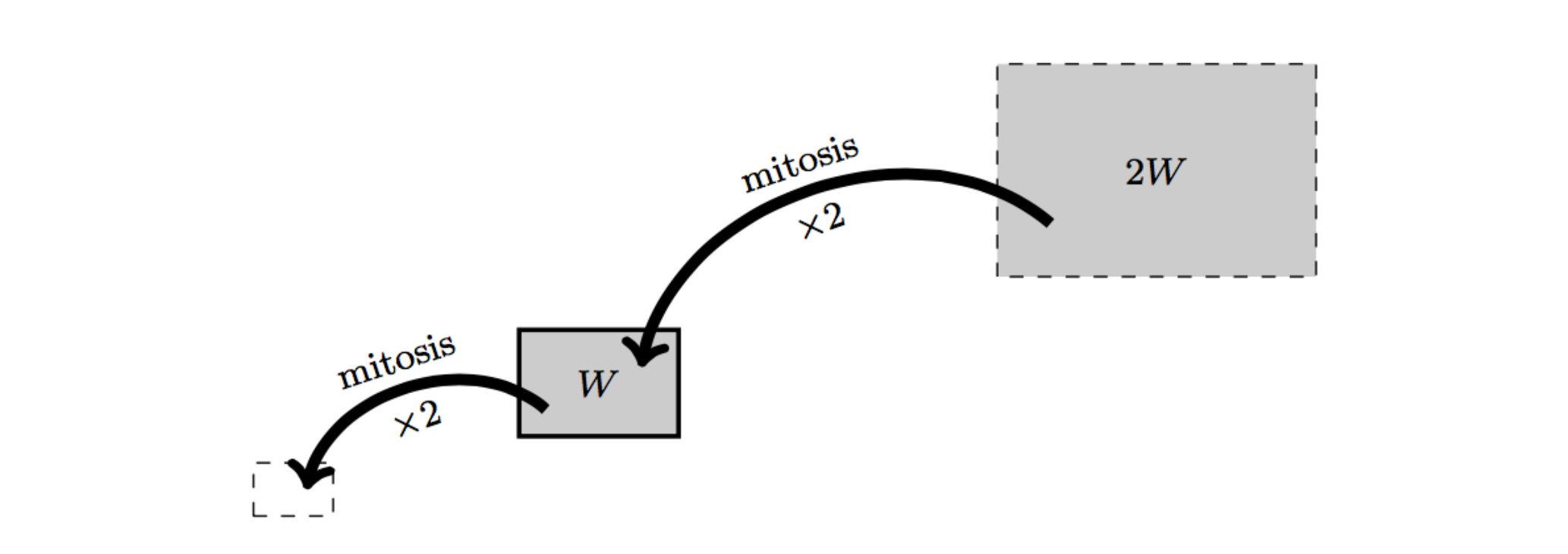}
\caption{Individuals leaving and entering the control volume $W \subset\P$ through mitosis.}
\label{fig:mitosis}
\end{figure}

As was done in \citet{Metz1986}, we impose the following
\begin{convention}\label{conv:istate}
If a transformed \istate argument falls outside $\P$ we shall assume that the term in which it occurs equals zero.
\end{convention}

Then, for an arbitrary control volume $W\subset\P$, the source of cells in $W$ is given by
\begin{align}
   \int\limits_W \!S(t,x,y) \dif y 
   &= 2 \!\!\int\limits_{2W}\!\! \Phi(\tilde y,\vect u(t,x,\tilde y)) c(t,x, \tilde y)\dif \tilde y 
      - \!\!\int\limits_W\!\! \Phi(y,\vect u(t,x,y)) c(t,x,y) \dif y\,. \label{equ:sourceTermIntegral}
   \intertext{Equation~\eqref{equ:sourceTermIntegral} is obvious if $W$, $2W$, and $\frac12 W$ are pairwise disjoint. For the general case with arbitrary $W$ we refer to the proof in Appendix~\ref{app:sourceTerm}.
   For the integral over~$2W$, we use the change of variables $\tilde y(y)=2y$, for which $\det(J_{\tilde y})=2^{p}$, and obtain}
   \int\limits_W\! S(t,x,y) \dif y 
   &= \!\!\int\limits_{W}\!\! 2^{p+1} \Phi(2y,\vect u(t,x,2y)) c(t,x,2y)\dif y 
     - \!\!\int\limits_W\!\! \Phi(y,\vect u(t,x,y)) c(t,x,y) \dif y\,. \nonumber
\end{align}
Since this holds for any control volume $W$, we get
\begin{align}
   S(t,x,y) = 2^{p+1} \Phi(2y,\vect u(t,x,2y)) c(t,x,2y) - \Phi(y,\vect u(t,x,y)) c(t,x,y)\,, \forall y\in \P. \label{equ:sourceTerm}
\end{align}

\paragraph{Spatial flux.}
The flux over the spatial boundary results from a combination of diffusion (random motion), chemotaxis (with respect to various free molecular species volume concentrations), and haptotaxis (with respect to the ECM density) of the structured cell population. 
Here we define the diffusion and taxis terms following \citet{Andasari2011, Gerisch2008} as
\begin{align}
\label{equ:spatialFlux}
   F(t,x,y) &= -D_c\nabla_x c + c (1-\rho(C,v))\left(\sum_{k=1}^{q} \chi_k \nabla_x m_k + \chi_v \nabla_x v \right)\,,
\end{align}
where the free molecular species with volume concentration~$m_k$ may either act as chemoattractants or as chemorepellents. 
We assume that the diffusion coefficient~$D_c(\cdot)$ as well as the taxis coefficients~$\chi_v(\cdot)$  and~$ \chi_k(\cdot)$, $k=1,\dots,q$, can, in particular, depend on the \istate~$y\in\P$.
More complex forms of~\eqref{equ:spatialFlux} are indeed conceivable and we provide an initial discussion in Section~\ref{sec:conclusion}.

\paragraph{Structural flux.}
The flux over the structural boundary represents changes in the \istate, that is changes in the surface concentration of bound molecules on the cells' surface, and thus results from binding and unbinding events of molecules to and from the cells' surface.

We assume that the binding rates of the free molecular species $m_1,...,m_p$ to the cell surface depend on the already bound molecules on the cell surface, i.e. the \istate~$y$, as well as on the available free molecules, i.e. the free molecular species volume concentration~$\vect m(t,x)$. 
Thus we denote the non-negative binding rate vector by $\vec{b}(y,\vect{m})\in\R^p$.
In contrast, we assume that the unbinding rates only depend on the \istate~$y$, which implies that unbinding is not restricted by~$\vect m(t,x)$.
Thus we denote the non-negative unbinding rate vector by $\vec{d}(y)\in\R^p$.
In summary, these binding and unbinding rates lead to an associated \emph{net binding rate} for the \istate~$y$ given by
$
   \vec{b}(y,\vect{m}) - \vect d(y)\,.
$
The net binding rate describes an amount of molecules bound per surface area per unit time, hence the unit of this rate is given by~$[\rm{(\upmu mol/cm^2)/s}]$.

Since the \istate space $\P$ is defined as the set of all admissible structure states, it is necessary that the net binding rate vector field does not point out of $\P$ on $\partial\P$, i.e. that
\begin{align}
 \left(\vect b(y,\vect m) - \vect d(y)\right) \normal(y) &\leq 0 \qquad \text{ for } t\in\I\,,\,\, x\in\overline{\D}\,,\,\, y\in\partial\P\,, \label{equ:velocityField-noCrossing}
\end{align}
where $\normal(y)$ denotes the outer unit normal vector on $\partial\P$ in $y\in\partial\P$.
This condition must be fulfilled by the particular choice of $\vect b$ and $\vect d$ in specific models.

Now the flux is given by the product of the structured cell density and the net binding rate, hence has the form
\begin{align}
   G(t,x,y) &= c(t,x,y) \bigl(\vec{b}(y,\vect{m}) - \vect d(y) \bigr)\,. \label{equ:structuralFlux}
\end{align} 
This form maintains the interpretation of the structural flux $G(t,x,y)$ as, for example, growth in size-structured populations \citep{Chapman2007,Tucker1988,Metz1986,Webb2008}.

\subsection{Extracellular matrix}\label{sec:ECM}

The extracellular matrix (ECM) consists of fibrous proteins such as collagen or vitronectin. These proteins are assumed to be static, i.e. we do not consider any transport terms for the ECM. The ECM is degraded by one or more of the free molecular species or the surface-bound reactants and is remodelled by the stroma cells present in the tissue (which are not modelled explicitly).
The equation for the ECM is then
\begin{align*}
   \frac{\dd}{\dd t} v(t,x) &= - \underbrace{\boldsymbol\delta_v^\T \vect r(t,x) v(t,x)}_{\text{degradation}} + \underbrace{\psi_v(t,\vect u(t,x))}_{\text{remodelling}}\,,
\end{align*}
where~$\boldsymbol\delta_v \in \R^{p+q}$ is the non-negative vector of ECM degradation rates and $\psi_v(t,\vect u)$ represents the remodelling term. To ensure non-negativity of the ECM density, we require $\psi_v(t,\vect u) \geq 0$ for $v=0$. A common formulation for the remodelling term is a constant rate together with a volume-filling term, see, e.g., \citet{Domschke2014},
\begin{align}
   \psi_v(t,\vect u(t,x)) &= \mu_v (1-\rho(C(t,x),v(t,x)))^+\,. \label{equ:ECMRemodelling}
\end{align}

\subsection{Molecular species}\label{sec:motileSpecies}

We assume that the free molecular species, as described by their volume concentrations~$m_i$, $i=1,\dots,q$, rearrange spatially driven by diffusion only. 
Furthermore, they are produced by either the cells directly or by chemical reactions. 
Potentially, some of the species undergo natural decay. 
The first~$p$ species may also bind to and unbind from the cell surface. 
All these effects can be captured in the following equation describing the dynamics of all free molecular species volume concentrations:
\begin{align}
\begin{split}
   \frac{\dd}{\dd t} \vect m(t,x) &= \underbrace{\nabla_x \cdot \left[\vect D_{\vect m} \nabla_x \vect m(t,x)\right]}_{\text{diffusion}}
            - \underbrace{\int\limits_\P \left(\hat{\vect b}(y,\vect m) - \hat{\vect d}(y)\right)s(t,x,y)\dif y}_{\text{binding/unbinding}} \\
            &\quad+ \underbrace{\boldsymbol \psi_{\vect m}(\vect u(t,x), \vect r(t,x))}_{\text{production}}
            - \underbrace{\diag(\boldsymbol \delta_{\vect m}) \vect m(t,x)}_{\text{decay}}\,.
\end{split}\label{equ:motileSpecies}          
\end{align}
In the above, 
$\vect D_{\vect m} = \diag(D_{m_1}, \dots, D_{m_q}) \in \R^{q,q}$ denotes the diagonal matrix containing the non-negative diffusion constants of the individual species.
Furthermore,~$\boldsymbol\psi_{\vect m}(\vect u, \vect r)$ is the vector of production terms, which depends on the structured cell and ECM densities and the free as well as bound molecular species volume concentrations.
This production term is in particular \istate-dependent, explicitly through $c$ and implicitly through $\vect n$, and thus provides influence of the structure on the dynamics of the overall system. 
This is a strong feature of our structured modelling framework and the necessity of such a feature provided the main initial motivation to consider a structured approach. 
The modelling examples in Sections~\ref{sec:numericalExample} and~\ref{sec:uPASystem} will highlight this in more detail.
In order to ensure the non-negativity of $\vect{m}$, we require, for $j=1,\dots,q$, that~$\left(\boldsymbol\psi_{\vect m}(\vect u, \vect r)\right)_j\ge0$ if $m_j=0$.
Next, the vector~$\boldsymbol \delta_{\vect m}$ contains the non-negative rates of decay of the individual species.
Finally, we discuss the reasoning behind the remaining binding/unbinding term in more detail below.

The rate of change of $\vect m$ due to binding or unbinding events to the cell surface is zero for the components $p+1, p+2,\dots,q$ since those do not bind to the cell surface. 
Thus we will derive the appropriate rate of change for the first $p$ components below.
For a unified treatment of all components, however,
we extend the binding and unbinding rate vectors by zeros, that is we define
\begin{align}
  \label{eqn:ext_bin_unbind_rate}
   \hat {\vect b}(y,\vect m) = \begin{pmatrix}  \vect b(y,\vect m) \\[1mm] \vect 0\end{pmatrix} \in \R^q 
   \quad\text{and}\quad \hat{\vect d}(y) = \begin{pmatrix}\vect d(y)\\ \vect 0\end{pmatrix} \in \R^q\,.
\end{align}
The rate of change of the volume concentration $\vect m$ due to binding or unbinding events to cell surfaces is the combined effect of the corresponding rates of change per \istate; thus the binding/unbinding term in~\eqref{equ:motileSpecies} is an integral over the \istate space.
The rate of change of the volume concentration $\vect m$ due to binding/unbinding to/from cell surfaces in \istate~$y$ can be seen as the product of the net binding rate $\hat{\vect b}(y,\vect m) - \hat{\vect d}(y)$,
which gives the amount of molecules being bound per surface area per unit time ($[\rm{(\upmu mol/cm^2)/s}]$), and the structured cell surface density $s(t,x,y)$, which denotes, per unit volume in space and per unit volume in the \istate, the surface area of the cells at $t$ and $x$ that have surface concentration $y$ ($[\rm{cm^{2}/(cm^3\cdot (\upmu mol/cm^2)^p)}]$).
This explains the integrand in the binding/unbinding term in~\eqref{equ:motileSpecies}.
Furthermore note that a positive component $j$ of the net binding rate means that the volume concentration $m_j$ decreases and thus the minus sign in front of the integral is required. 
Finally, observe that the earlier conditions on the binding rate vector $\vect b$ ensure non-negativity of $\vect m$.

\subsection{Summary of the model, non-dimensionalisation, initial and boundary conditions}

For the convenience of the reader, we summarize below the equations of the structured model for the structured cell density, the ECM densitiy, and the free molecular species volume concentrations as they have been derived in Sections~\ref{sec:generalCellPopulation},~\ref{sec:ECM}, and~\ref{sec:motileSpecies}, respectively:
\begin{subequations}
\begin{align}
      \begin{split}
         \frac{\dd c}{\dd t} &= \nabla_x \cdot \left[D_c\nabla_x c 
         - c (1-\rho(C,v))\left(\sum_{k=1}^q \chi_k \nabla_x m_k + \chi_v \nabla_x v\right)\right] \\[1mm]
     &\qquad  - \nabla_y\cdot\left[\left(\vect b(y,\vect m) - \vect d(y)\right) c\right] \\[2mm]
      &\qquad + 2^{p+1} \Phi(2y,\vect u(t,x,2y)) c(t,x,2y)  - \Phi(y,\vect u(t,x,y)) c(t,x,y)\,,
      \end{split} \label{equ:generalModel-cells}\\[5mm]
     \frac{\dd v}{\dd t} &= -\boldsymbol\delta_v^\T \vect r v + \psi_v(t,\vect u)\,, \label{equ:generalModel-ecm}\\[5mm]
     \begin{split}
     \frac{\dd \vect m}{\dd t} &= \nabla_x \cdot \left[\mat D_{\vect m} \nabla_x \vect m\right] - 
            \int\limits_\P \left(\hat{\vect b}(y,\vect m) - \hat{\vect d}(y)\right) s \dif y \\[1mm]
            &\qquad + \boldsymbol \psi_{\vect m}(\vect u,\vect r) - \diag(\boldsymbol \delta_{\vect m}) \vect m\,.
     \end{split}\label{equ:generalModel-motile}
\end{align}\label{equ:generalModel}
\end{subequations}

\noindent
In the above, we have suppressed the arguments $(t,x)$ and $(t,x,y)$ except in the proliferation term in Eq.~\eqref{equ:generalModel-cells} where it is necessary to show its dependence on $2y$.

We non-dimensionalise system~\eqref{equ:generalModel} by using the following dimensionless quantities
\begin{align}
\begin{aligned}
   \tilde t &= \frac{t}{\tau},\
 & \tilde{x} &= \frac{x}{L},\
 & \tilde{y} &= \frac{y}{y_\ast},\ \\
   \tilde{c}(\tilde{t},\tilde{x},\tilde{y}) &= \frac{c(t,x,y)}{c_\ast},\ 
 & \tilde{v}(\tilde{t},\tilde{x}) &= \frac{v(t,x)}{v_\ast},\  
 & \tilde{\vect m}(\tilde{t},\tilde{x}) &= \frac{\vect m(t,x)}{m_\ast}. 
\end{aligned}   \label{equ:scalings}
\end{align}
The scaling parameters are given in Appendix~\ref{app:ParameterTables} and the appropriate non-dimensionalised model parameters are collected there in Table~\ref{tab:parameters}. The units and non-dimensionalisation of intermediate quantities are shown in Table~\ref{tab:intermediateQuantities}. 
With the scalings defined in~\eqref{equ:scalings}, the system obtained by non-dimensionalisation of~\eqref{equ:generalModel} looks identical to the original one, but with a tilde on each quantity.
For notational convenience we will omit the tilde signs in the following but will consider system~\eqref{equ:generalModel} as the non-dimensionalised system and always refer to non-dimensionalised quantities.

System~\eqref{equ:generalModel} is supposed to hold for $t\in\I$, $x\in\D$ and $y\in\P$ and is completed by 
initial conditions
\begin{align}
 c(0,x,y) &= c_0(x,y)\,, & v(0,x) &= v_0(x)\,, & \vect{m}(0,x) &= \vect{m}_0(x) & \text{ for } x\in\overline\D, y\in\overline\P\,,
\end{align}
and zero-flux boundary conditions in space, that is
\begin{align}\label{eqn:BC-space}
\begin{split}
     \left[D_c\nabla_x c 
         - c (1-\rho(C,v))\left(\sum_{k=1}^q \chi_k \nabla_x m_k + \chi_v \nabla_x v\right) 
         \right] \cdot \normal(x) = 0\,,&\\
   \left[\mat D_{\vect m} \nabla_x \vect m\right] \cdot \normal(x) = \vect 0\,,&\\
    \text{ for } t\in\I\,,\,\, x\in\partial\D\,,\,\, y\in\overline{\P}\,,&
\end{split} 
\end{align}
where~$\normal(x)$ denotes the unit outer normal vector on~$\partial \D$ in~$x\in\partial\D$. 

Since the equation for the structured cell density~\eqref{equ:generalModel-cells} is hyperbolic in the \istate variable, we can only impose boundary conditions on the inflow boundary part of~$\P$, i.e., where $\left[\vect{b}(y,\vect{m})-\vect{d}(y)\right]\cdot \normal(y)< 0$ holds. Here, $\normal(y)$ denotes the unit outer normal vector on~$\partial \P$ in~$y\in\partial\P$. 
Clearly, the inflow boundary part of~$\P$ may change with $(t,x)$ through changes in $\vect m(t,x)$ and is thus denoted and defined by
\begin{align}
\partial\P_{in}(t,x):=\left\{y\in\partial\P:\left[\vect{b}(y,\vect{m}(t,x))-\vect{d}(y)\right]\cdot \normal(y)< 0\right\}\,.
\label{equ:inflowBoundary}
\end{align}
Since we assume that no cells with \istate{s}  
outside~$\P$ exist, we impose a zero Dirichlet boundary condition on the inflow boundary of the \istate space, that is
\begin{align}
   c(t,x,y) &= 0 \qquad \text{ for } t\in\I, x\in\overline{\D}, y\in\partial\P_{in}(t,x)\,. \label{equ:BC-iState}
\end{align}
Recall that, according to our modeling, cells in \istate $y\in\P$ divide into cells in \istate $y/2\in\P$ since $\P$ is convex with accumulation point $0$. Thus the proliferation term in the structured cell density equation does not create cells on the boundary of $\P$ and is thus consistent with the above zero Dirichlet boundary condition on $\partial\P_{in}(t,x)$.

On the part of $\partial\P$, where we do not have an inflow situation, i.e. where we cannot prescribe boundary conditions, the flux in outer normal direction is zero, which follows directly from \eqref{equ:velocityField-noCrossing}. On the inflow boundary $\partial\P_{in}$, where we impose zero Dirichlet boundary conditions, the flux in outer normal direction is also zero and hence, for the whole boundary of the \istate space~$\P$ it holds that 
\begin{align}
 \left[\left(\vect b(y,\vect m) - \vect d(y)\right) c\right] \normal(y) &= 0 \qquad \text{ for } t\in\I\,,\,\, x\in\overline{\D}\,,\,\, y\in\partial\P\,. \label{equ:zeroFlux-iState}
\end{align}
We provide a more detailed discussion of these boundary conditions in the presentation of the specific models in Sections~\ref{sec:numericalExample} and~\ref{sec:uPASystem}.

\subsection{Derivation of a non-structured model corresponding  to~\eqref{equ:generalModel}}
\label{sec:generalDerivationNonStructuredModel}

The total cell density $C(t,x)$ is obtained by integrating the structured cell density $c(t,x,y)$ over the \istate-space $\P$.
The aim of this section is to take the structured model~\eqref{equ:generalModel} as a starting point and to derive a suitable, corresponding non-structured model. 
That model will be formulated exclusively in terms of the non-structured quantities $C(t,x)$, $v(t,x)$, and $\vect m(t,x)$. 
Please note that $v(t,x)$, and $\vect m(t,x)$ in the non-structured model will not be identical with the variables of the same name in the structured model because their defining equations will be different since structured terms need to be approximated by non-structured ones.
However, their principle meaning will be the same and thus we chose to also stick with the same variable names.

In the derivation of the non-structured model below it is necessary to approximate terms involving structured expressions with expressions which involve only variables of the non-structured model.
For our purposes here, this will be achieved, in general, by replacing structured terms  by their \istate mean as well as the structured cell density by its mean value with respect to the \istate{} space;
higher-order approximations of the latter are of course possible and we comment on these in the conclusion in Section~\ref{sec:conclusion}.
In order to proceed, we first define the mean structured cell density and the centre of mass of the \istate space $\P$ by,
\begin{equation*}
  \bar{c}(t,x):=\frac{1}{|\P|}\int_\P c(t,x,y) \dif y = \frac{1}{|\P|} C(t,x)
  \quad\text{and}\quad
  \bar y = \frac{1}{|\P|}\int_\P y\dif y\,,
\end{equation*}
respectively.

The parameters $D_c(y)$, $\chi_k(y)$ for $k=1,\dots,1$, and $\chi_v(y)$ of the spatial flux expression~\eqref{equ:spatialFlux} are replaced by their mean values over the \istate space. 
These constants are denoted by $\bar D_c$, $\bar \chi_k$, and $\bar \chi_v$, respectively.
Also, the (extended) binding and unbinding rate vectors, $\hat{\vect{b}}(y,\vect m)$ and $\hat{\vect{d}}(y)$, respectively, see~\eqref{eqn:ext_bin_unbind_rate}, are replaced by their \istate-means, which are denoted by~$\overline{\hat{\vect{b}}}(\vect m)$ and~$\overline{\hat{\vect{d}}}$, respectively.

The situation is different and more involved in, for example, the bound molecular species volume concentrations $\vect n$, since its defining expression depends on the \istate $y$ explicitly but also implicitly through the structured cell density $c(t,x,y)$.
In this case we replace $c(t,x,y)$ by its \istate mean $\bar{c}(t,x)$ and obtain the following approximation
\[
 \vect{n}(t,x) 
          = \int_{\P} y \varepsilon c(t,x,y) \dif y
          \approx \varepsilon C(t,x) \frac{1}{|\P|} \int_{\P} y \dif y
          = \varepsilon \bar y C(t,x)
          =: \bar{\vect{n}} (t,x)\,.
\]
The new quantity $\bar{\vect{n}} (t,x)$ is computable from non-structured quantities and can thus be used in the non-structured model.
We are now in the position to introduce the following non-structured versions of $\vect u$ and $\vect r$
\[
  \bar{\vect{u}}(t,x) := \begin{pmatrix}\bar{c}(t,x) \\ v(t,x) \end{pmatrix}
  \quad\text{and}\quad
  \bar{\vect{r}}(t,x) := \begin{pmatrix}\bar{\vect{n}}(t,x) \\ \vect{m}(t,x) \end{pmatrix}\,.
\]

We can now further approximate the proliferation rate $\Phi(y,\vect u)$ as follows
\[
  \Phi(y,\vect u) \approx \Phi(y,\bar{\vect u}) \approx \bar\Phi(\bar{\vect u})\,,
\]
where the first approximation is the replacement of $c(t,x,y)$ by $\bar c(t,x)$ and the second approximation (which might be exact) is the determination of the \istate mean of $\Phi(y,\bar{\vect u})$.
In a similar fashion we arrive at the approximation $\bar \psi_v(t,\bar{\vect u})$ for the remodelling term $\psi_v(t,\vect u)$ in the ECM density equation~\eqref{equ:generalModel-ecm} and 
at the approximation $\bar{\boldsymbol \psi}_{\vect m}(\bar{\vect{u}},\bar{\vect{r}})$ for the production term $\boldsymbol \psi_{\vect m}(\vect u, \vect r)$ in the free molecular species volume concentration equation~\eqref{equ:generalModel-motile}.

With all the above preparatory definitions and approximations at hand, we now derive the non-structured model and start by integrating the structured cell density, i.e. Eq.~\eqref{equ:generalModel-cells}, over the \istate space.
Under the assumption that we can exchange integration and differentiation on the left-hand side, i.e. that we can apply Leibniz's rule for differentiation under the integral sign \citep{Halmos1978}, we obtain
\begin{align*}
         \frac{\dd C}{\dd t} &=\!\! \int\limits_\P \!\Biggl(\!\nabla_x \!\cdot\! \left[D_c\nabla_x c 
         - c (1-\rho(C,v))\!\left(\sum_{k=1}^q \!\chi_k \nabla_x m_k \!+ \chi_v \nabla_x v\!\right)\!\right] \Biggr)\dif y \\[1mm]
     &\qquad  - \int\limits_\P\left( \nabla_y\cdot\left[\left(\vect b(y,\vect m) - \vect d(y)\right) c\right] \right)\dif y \\[2mm]
      &\qquad + \int\limits_\P\left( 2^{p+1} \Phi(2y,\vect u(t,x,2y)) c(t,x,2y)  - \Phi(y,\vect u(t,x,y)) c(t,x,y)\right)\dif y\,.
\end{align*}

Since, according to~\eqref{equ:zeroFlux-iState}, we have that the flux is zero in outer normal direction on the boundary of the \istate space, the second integral on the right-hand side vanishes using the divergence theorem.

Furthermore, cf. Equation~\eqref{equ:sourceTermIntegral} on page~\pageref{equ:sourceTermIntegral}, using the change of variables $\tilde y(y) = 2y$ in the first half of the integral over the proliferation term and upon immediately dropping the tilde-sign and invoking Convention~\ref{conv:istate}, we arrive for this integral at
\begin{align*}
  \int\limits_\P &\left( 2^{p+1} \Phi(2y,\vect u(t,x,2y)) c(t,x,2y) - \Phi(y,\vect u(t,x,y)) c(t,x,y)\right)\dif y
  \\
  &\qquad =\int\limits_\P\Phi(y,\vect u(t,x,y)) c(t,x,y)\dif y\,,
  \\
  \intertext{and finally, replacing the structured proliferation rate $\Phi(y,\vect u)$ by its \istate-independent approximation $\bar \Phi(\bar{\vect u})$, we obtain}
  &\qquad \approx \bar\Phi(\bar{\vect u})C\,.
\end{align*}

\begin{subequations}\label{equ:generalNonStructuredModel}

Replacing the remaining \istate-dependent parameter functions $D_c(y)$, $\chi_k(y)$ for $k=1,\dots,1$, and $\chi_v(y)$ in the equation for $C$ by their respective \istate-independent approximations $\bar D_c$, $\bar \chi_k$, and $\bar \chi_v$, and applying again Leibniz's rule for differentiation under the integral sign, we obtain the following non-structured equation for the total cell density
\begin{align}
         \frac{\dd C}{\dd t} &= \nabla_x \!\cdot\! \left[\bar D_c\nabla_x C 
         - C (1-\rho(C,v))\left(\sum_{k=1}^q \bar\chi_k \nabla_x m_k + \bar\chi_v \nabla_x v\right)\right]\! + \bar\Phi(\bar{\vect u}) C\,.
\end{align}

We now turn to derive the non-structured counterpart of Eq.~\eqref{equ:generalModel-ecm}, the equation for the ECM density.
Making use of the non-structured approximations $\bar{\vect{r}}$ and $\bar\psi_v(t,\bar{\vect{u}})$ we can simply write it down as
\begin{align}
  \frac{\dd v}{\dd t} &= -\boldsymbol\delta_v^\T \bar{\vect{r}} v + \bar\psi_v(t,\bar{\vect{u}})\,.
\end{align}

Finally, we derive the non-structured equation for the bound molecular species volume concentrations and take Eq.~\eqref{equ:generalModel-motile} as starting point. 
For the production term we use the earlier discussed approximation $\bar{\boldsymbol \psi}_{\vect m}(\bar{\vect{u}},\bar{\vect{r}})$ as replacement.
The term for the concentration changes due to surface binding and unbinding is approximated as follows
\begin{align*}
    -\int\limits_\P \left(\hat{\vect b}(y,\vect m) - \hat{\vect d}(y)\right) \varepsilon c(t,x,y) \dif y
  &\approx-\int\limits_\P \left(\hat{\vect b}(y,\vect m) - \hat{\vect d}(y)\right) \varepsilon \bar{c}(t,x) \dif y
  \\
  &=-\varepsilon C(t,x) \frac{1}{|\P|} \int\limits_\P \left(\hat{\vect b}(y,\vect m) - \hat{\vect d}(y)\right) \dif y
  \\
  &=-\varepsilon C(t,x)\left(\overline{\hat{\vect{b}}}(\vect m) - \overline{\hat{\vect{d}}}\right)\,.
\end{align*}
Thus, taking that all together, we arrive at
\begin{align}
     \frac{\dd \vect m}{\dd t} &= \nabla_x \cdot \left[\mat D_{\vect m} \nabla_x \vect m\right] - 
            \left(\overline{\hat{\vect{b}}}(\vect m) - \overline{\hat{\vect{d}}}\right) \varepsilon C 
            + \bar{\boldsymbol \psi}_{\vect m}(\bar{\vect{u}}, \bar{\vect{r}}) - \diag(\boldsymbol \delta_{\vect m}) \vect m\,.
\end{align}

Finally, the initial and boundary conditions of the structured model give rise to the following initial conditions
\begin{align}
 C(0,x) &= \int_\P c_0(x,y)\dif y\,, & v(0,x) &= v_0(x)\,, & \vect{m}(0,x) &= \vect{m}_0(x) & \text{ for } x\in\overline\D\,,
\end{align}
and zero-flux boundary conditions
\begin{equation}
\begin{aligned}
 \left[\bar D_c\nabla_x C 
       - C (1-\rho(C,v))\left(\sum_{k=1}^q \bar\chi_k \nabla_x m_k + \bar\chi_v \nabla_x v\right)
 \right]
 \cdot \normal(x) = 0\,,&
  \\
  \left[\mat D_{\vect m} \nabla_x \vect m\right] \cdot \normal(x) = \vect 0\,,&
  \\
  \text{ for } t\in\I\,,\,\, x\in\partial\D\,,
\end{aligned}
\end{equation}
in the non-structured case.

\end{subequations}



\section{A Generic Structured Model of Cancer Invasion}\label{sec:numericalExample}

Now that we have derived the general structured-population model \eqref{equ:generalModel}, we want to explore the influence of the structure on the spatio-temporal dynamics of the model components starting with a very simple model.

In cancer modelling, most approaches exploring molecular-cell population dynamic interactions are either based on spatio-temporal PDEs of reaction-diffusion-taxis type, \citep{Gatenby1996,Anderson2000,Byrne2004,Chaplain2005,Domschke2014}, or continuum-discrete hybrid systems \citep{Anderson1998,Anderson2000,Anderson2005}, or more recently multiscale continuum models \citep{Ramis-Conde2008,Marciniak-Czochra2008,Macklin2009,Deisboeck2011,Trucu2013}. Of particular interest in cancer invasion is the interaction between the tumour cell population and various 
proteolytic enzymes, such as MMPs \citep{Parsons1997} or the uPA system \citep{Andreasen1997,Andreasen2000,Pepper2001} that enable the degradation of extracellular matrix components, thus promoting further local tumour progression. While the modelling of this interaction has already received a special attention \citep{Chaplain2005,Chaplain2006,Andasari2011,Deakin2013}, the structural characteristics of, e.g., the binding process of uPA to its surface receptor uPAR and the activation of matrix-degrading enzymes (MDEs) coupled with their simultaneous effects on cell motility and proliferation so far have been unexplored.

In this example, besides a structured population of cancer cells with cell density~$c$ and the ECM with density~$v$, we have two molecular species with volume concentrations~$m_1$ and~$m_2$. 
Cancer cells rearrange spatially through random motility, chemotaxis with respect to $m_1$, and haptotaxis with respect to $v$.
The first molecular species, $m_1$, is produced by the cancer cells and can bind to the surface of the cells; the latter process gives rise to the \istate of a cell. 
The second molecular species, $m_2$, is solely produced (or activated from an abundantly present inactive form) through the action of bound molecules of the first type and subsequently degrades the ECM. 

To set up the model, note that only the first molecular species, $m_1$, binds to the cell surface and thus we consider the one-dimensional \istate space $\P=(0,Y)$, where $Y>0$ denotes the maximum surface density of the first species. We detail our assumptions regarding the coefficients and parameters of the general model~\eqref{equ:generalModel} in the following paragraphs.

For the structured cell density equation, i.e. Eq.~\eqref{equ:generalModel-cells}, we assume the diffusion coefficient $D_c$ and the chemotaxis coefficient $\chi_1$ to be constant, $\chi_2=0$, and the haptotaxis coefficient~$\chi_v(y)$ to be \istate-dependent.
The proliferation rate~$\Phi$ of the cancer cells is considered to be restricted by spatial constraints and to be \istate-independent and takes the form
\begin{align}
	\Phi(y,\vect{u})\equiv \Phi(C,v) = \mu_c (1-\rho(C,v))\,. \label{equ:proliferationRate}
\end{align}
Finally, for the binding rate $\vect b(y,\vect m)$ we assume that it is proportional to the available free molecular volume concentration $m_1$ and also proportional to the free capacity of the cell's surface, i.e. $Y-y$,
and for the unbinding rate $\vect d(y)$ we assume that it is proportional to the bound molecular surface density $y$. 
This gives rise to the following scalar rates
\begin{align}\label{equ:exampleBindingUnbindingRates}
  \vect b(y,\vect m) = (Y-y)\beta m_1
  \quad\text{and}\quad
  \vect d(y) = y\delta_y\,.
\end{align}

Turning to the ECM density equation~\eqref{equ:generalModel-ecm}, recall that the combined vector of bound and free molecular species volume concentrations is given by $\vect r = (n_1, m_1, m_2)^T \in \R^{3}$ and that ECM is degraded upon contact with $m_2$. We assume a constant ECM degradation rate $\delta_v$ and the vector of degradation rates has the form
\begin{equation}
\boldsymbol\delta_v^\T=(0,0,\delta_v)\,.
\end{equation}
The remodelling term is defined independent of $y$ and following Eq.~\eqref{equ:ECMRemodelling} as
\begin{equation}
  \psi_v(t,\vect u) \equiv \psi_v(C,v) = \mu_v(1-\rho(C,v))^+\,.
\end{equation}

Finally, in Eq.~\eqref{equ:generalModel-motile} for $\vect{m}$, we consider the following linear production and degradation terms with constant coefficients
\begin{align}
\boldsymbol\psi_{\vect m} =\begin{pmatrix} \alpha_{m_1}C\\[0.2cm] \alpha_{m_2}n_1 \end{pmatrix} \qquad\text{and}\qquad
\boldsymbol \delta_{\vect m} =\begin{pmatrix}\delta_{m_1} \\[0.2cm] \delta_{m_2} \end{pmatrix}\,.
\end{align}
Observe that the production term of $m_2$ depends implicitly on the \istate via the bound molecular species volume concentration~$n_1$.

These considerations lead to the following structured system
\begin{subequations}
\begin{align}
\begin{split}
\frac{\dd c}{\dd t} &= \nabla_x\cdot\left[ D_c \nabla_x c - c(1-\rho(C,v))\left(\chi_1 \nabla_x m_1 + \chi_v(y)\nabla_x v\right)\right]\\[1mm]
& ~- \nabla_y \!\cdot\! \left[ \bigl(\vect b(y,\vect m)\!-\vect d(y)\bigr) c \right]  
\!+ \Phi(C,v)\left[ 4c(t,x,2y) - c(t,x,y) \right],
\end{split} \label{equ:structuredExample-c}\\[2mm]
 \frac{\dd v}{\dd t} &= - \delta_v m_2 v + \mu_v (1-\rho(C,v))^+\,, \label{equ:structuredExample-v}\\[2mm]
 \frac{\dd m_1}{\dd t} &= \nabla_x\cdot \left[ D_{m_1} \nabla_x m_1 \right] \!-\! \Bigl(\!\bigl(Y\!\varepsilon C \!-\! n_1 \bigr)\beta m_1 \!-\! \delta_y n_1\Bigr) 
 \!+ \alpha_{m_1} C - \delta_{m_1} m_1\,, \label{equ:structuredExample-m1}
 \\[2mm]
 \frac{\dd m_2}{\dd t} &= \nabla_x\cdot \left[ D_{m_2} \nabla_x m_2 \right] + \alpha_{m_2} n_1 - \delta_{m_2} m_2\,. 
 \label{equ:structuredExample-m2}
\end{align}\label{equ:structuredExample}
\end{subequations}

As in the general model, system~\eqref{equ:structuredExample} is supposed to hold for $t\in\I$, $x\in\D$ and $y\in\P$ and is completed by initial conditions and appropriate boundary conditions. In space we have zero-flux boundary conditions, while in the \istate space $\P$, due to the hyperbolic nature of the structured cell equation in the \istate variable $y$, we have to determine the inflow boundary $\partial\P_{in}(x,t)$ as defined in~\eqref{equ:inflowBoundary}.
In this example, the \istate space is given by the interval~$(0,Y)$, hence the boundary is $\partial\P=\{0,Y\}$. With the corresponding outer unit normal vectors $\normal(0)=-1$ and $\normal(Y)=1$, we obtain
\begin{align*}
 (\vect{b}(0,\vect{m})-\vect d(0))\cdot \normal(0) &= -Y\beta m_1 \leq 0\,,\\
 (\vect{b}(Y,\vect{m})-\vect d(Y))\cdot \normal(Y) &= -Y\delta_y \leq 0\,.
\end{align*}
Provided that there exist molecules of the first type (meaning that the volume concentration $m_1(t,x)$ is positive) and 
that binding and unbinding may take place (meaning that the binding and unbinding rate parameters $\beta$ and $\delta_y$ are positive),
both terms above are negative and the inflow boundary is $\partial\P_{in}(t,x)=\{0,Y\}$. 
If $m_1(t,x)$ or one of the parameters is zero, for example if we consider that no unbinding occurs, then the corresponding term above is zero and we do not have a classical inflow boundary at the corresponding location. 
Still, in such a situation also the corresponding flux across the boundary is zero and we make computational use of such a zero-flux boundary condition in our numerical scheme.

As derived in Section~\ref{sec:generalDerivationNonStructuredModel} and using the mean value $\bar\chi_v$ of the \istate-dependent coefficient $\chi_v(y)$ and the mean values
\begin{align*}
 \bar y &= \frac{Y}{2}\,,  &
 \bar{\vect{b}}(\vect m) &=  (Y\!-\bar y)\beta m_1\,,\quad\text{ and} &
 \bar{\vect{d}} &= \bar y \delta_y\,,
\end{align*}
we obtain the following corresponding non-structured model
\begin{subequations}
\begin{align}
\frac{\dd C}{\dd t} &= \nabla_x\!\cdot\!\left[ D_c \nabla_x C 
	\!-\! C(1\!-\!\rho(C,v))\left(\chi_1 \nabla_x m_1 + \bar\chi_v\nabla_x v \right)\right] + \Phi(C,v) C\,,\label{equ:non-structuredExample-c}\\
\frac{\dd v}{\dd t} &= - \delta_v m_2 v + \mu_v (1-\rho(C,v))^+\,, \label{equ:non-structuredExample-v}\\	
\frac{\dd m_1}{\dd t} &= \nabla_x\cdot \left[ D_{m_1} \nabla m_1 \right]  - \left((Y-\bar y)\beta m_1 - \bar y \delta_y \right)\varepsilon C + \alpha_{m_1}C - \delta_{m_1} m_1 \,,\label{equ:non-structuredExample-m1}\\
\frac{\dd m_2}{\dd t} &= \nabla_x\cdot \left[ D_{m_2} \nabla m_2 \right] + \alpha_{m_2}\varepsilon \bar y C - \delta_{m_2} m_2 \,.\label{equ:non-structuredExample-m2}
\end{align}\label{equ:non-structuredExample}
\end{subequations}

\subsection{Numerical simulations of the structured and corresponding non-structured model}
The simulations in this section highlight the difference that the structural binding information makes in characterising the dynamics in the structured case \eqref{equ:structuredExample} versus the corresponding non-structured system \eqref{equ:non-structuredExample}. 

In these numerical simulations, we use the following basic parameter set~$\paramSet$  similar to that used in \citet{Domschke2014}:
\begin{align}
\begin{aligned}
c:&& D_c &= 10^{-4}  & \chi_v &= 0.05 & \chi_1 &= 0.001  & \mu_c &= 0.1 \\[-0.5mm]
\istate:&& Y&=1 & \varepsilon &= 0.1 & \beta &= 0.5 & \delta_y &= 0\\[-0.5mm]
v:&& && \delta_v &= 10 & \mu_v &= 0.05  \\[-0.5mm]
m_1:&& D_{m_1}  &= 10^{-3} & \alpha_{m_1} &=0.1 & \delta_{m_1} &= 0.1 \\[-0.5mm]
m_2:&& D_{m_2} &= 10^{-3} & \alpha_{m_2} &=0.5 & \delta_{m_2} &= 0.1
\end{aligned}\tag{$\paramSet$}\label{equ:paramSet}
\end{align}

%
%
For the structured case, we consider also the haptotaxis coefficient~$\chi_v$ as a function linearly decaying from a maximal value $\chi_v^+$ for $y=0$ to a minimal value $\chi_v^-$ for $y=Y$ in order to model the effect that free receptors may bind to the ECM and thus accelerate haptotactic movement. Apposite to the basic parameter set~\ref{equ:paramSet}, we thus choose
\begin{align}
\chi_v(y)=  (\chi_v^- - \chi_v^+)\frac{y}{Y} + \chi_v^+\,,
\label{equ:linearHapto}
\end{align}
with $\chi_v^- = 0.001$ and $\chi_v^+ = 0.099$. This leads to the mean haptotactic coefficient $\overline{\chi}_v = 0.05$ for the non-structured model, which is identical with $\chi_v$ from the basic parameter set~$\paramSet$.

To complete the system, we choose the following initial conditions. The cancer cells are assumed to form a cancerous mass located at the origin (in $x$) with initially some molecules already bound to the cells' surfaces
\begin{align*}
 c_0(x,y) &= \exp(-(x^2+4(y-0.25)^2)/0.01)\,.
\intertext{This leads to the total initial cancer cell density}
 C_0(x) &= \int_\P c_0(x,y) \dif y \approx 0.0886 \exp(-x^2/0.01)\,,
\intertext{which we coose as initial cancer cell density for the corresponding non-structured model. The initial ECM density is chosen according to the spatial constraints}
 v_0(x) &= 1 - C_0(x)\,,
\intertext{and finally, we assume that the cancer cells already released some of the molecular species $m_1$ into the environment and set the initial molecular species volume concentrations to}
 \vect m_0(x) &= (0.5 C_0(x), 0)^\T\,.
\end{align*}

The results shown in Figures~\ref{fig:structuredExample} and \ref{fig:structuredExample-withUnbinding}, are obtained from simulations of the structured model~\eqref{equ:structuredExample} using the parameter set~\ref{equ:paramSet} with modifications as detailed in each figure caption. They present the structured cancer cell density~$c(t,x,y)$ in the \emph{spatio-structural} space~$\D\times\P$ in the top row and, in the bottom row, the total cancer cell density~$C(t,x)$, ECM density~$v(t,x)$, and bound and free molecular species volume concentrations~$\vect r(t,x) = (n_1(t,x), m_1(t,x),m_2(t,x))^\T$ in the spatial domain~$\D$ at initial time $t = 0$ and at times $t = 50,100,150$, and $200$ (from left to right). 

\begin{figure}
\begin{subfigure}[b]{\textwidth}
\includegraphics[width=\textwidth]{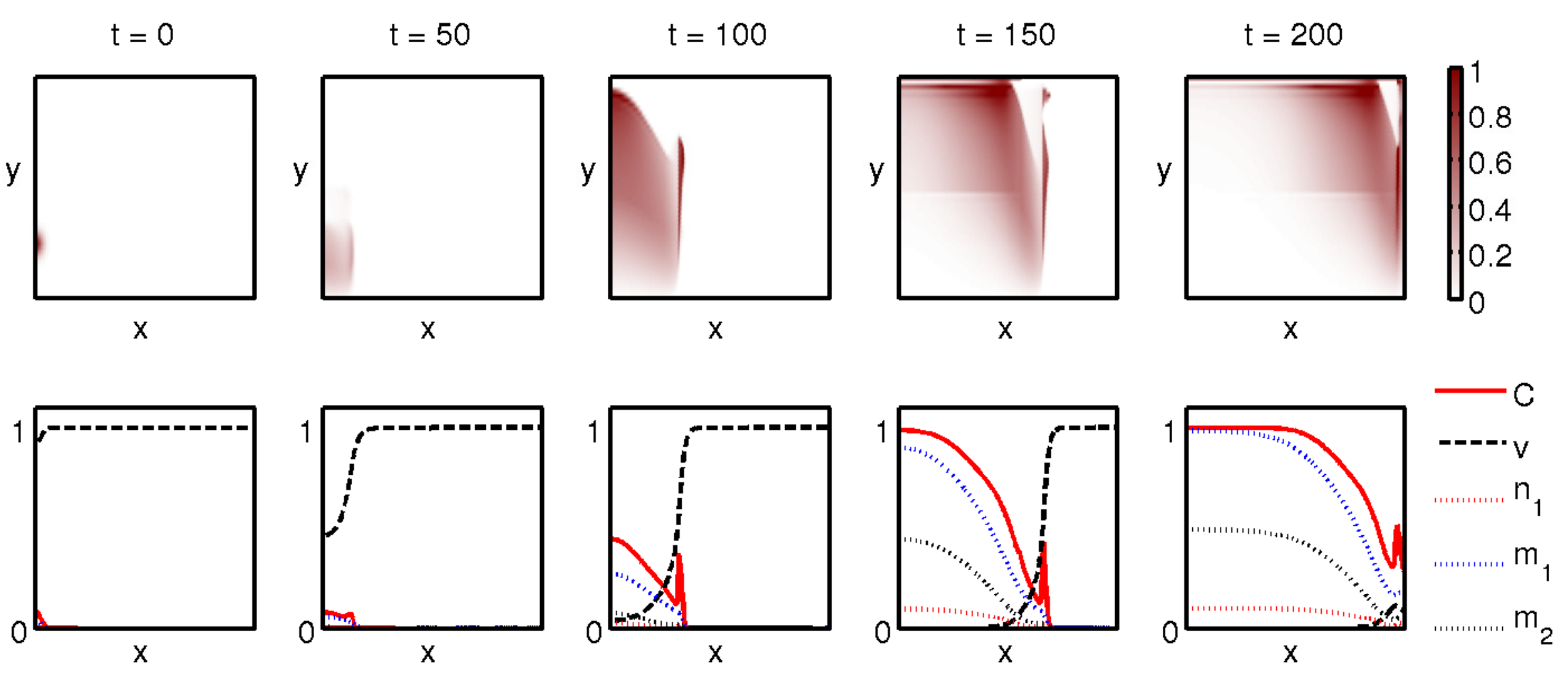}
\caption{
constant $\chi_v$}
\label{fig:structuredExample-constantHapto}
\end{subfigure}\\
\begin{subfigure}[b]{\textwidth}
\includegraphics[width=\textwidth]{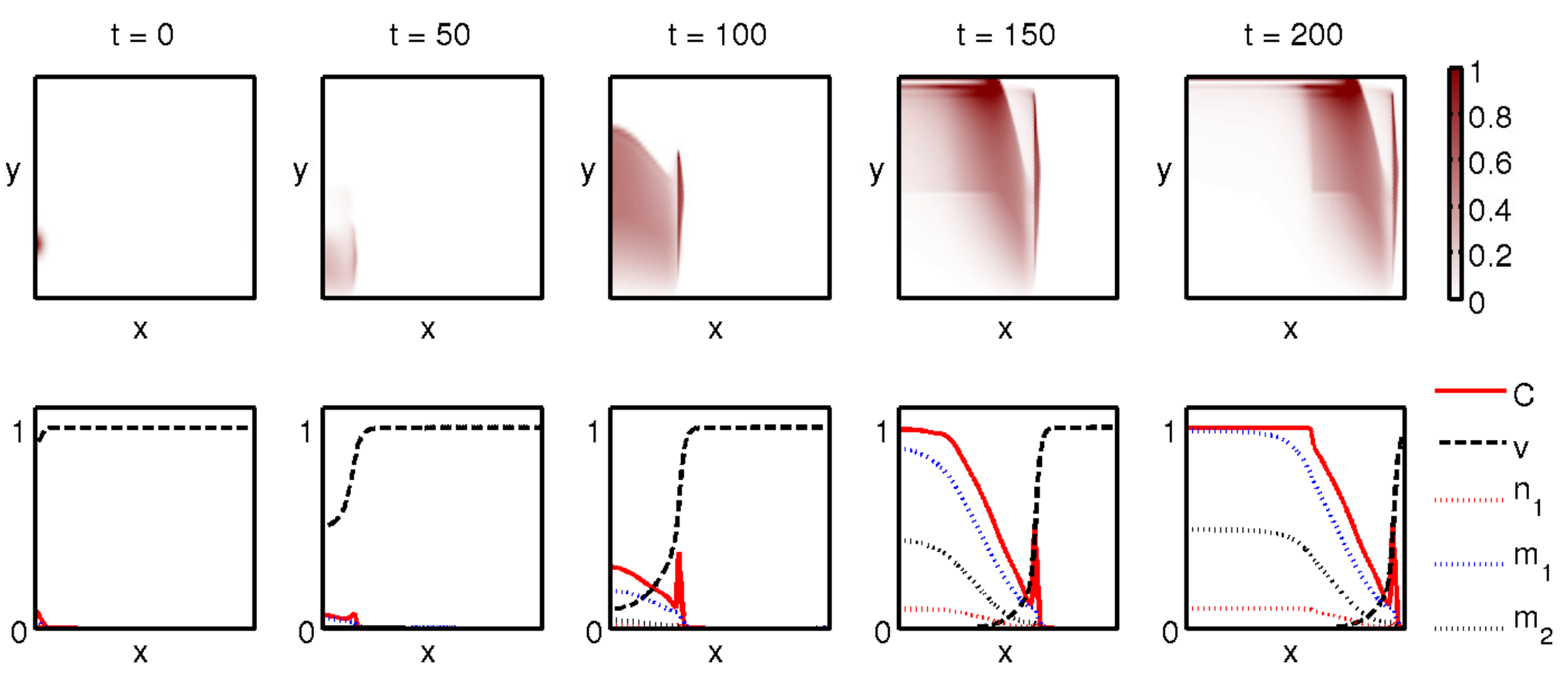}
\caption{
\istate-dependent $\chi_v$, cf.~\eqref{equ:linearHapto}}
\label{fig:structuredExample-linearHapto}
\end{subfigure}
\caption{Plots showing the computational simulation results at increasing time points (left to right) of the structured system~\eqref{equ:structuredExample} using the basic parameter set~\ref{equ:paramSet} with the haptotaxis term $\chi_v$ as specified in (a) and (b).
  The top row in each of (a) and (b) shows the evolution of the structured cancer cell density in the spatio-structural space $\D\times\P$; the bottom row in each of (a) and (b) shows the evolution of all non-structured variables in the spatial domain~$\D$.}
\label{fig:structuredExample}
\end{figure}

The results shown in Figure~\ref{fig:nonStructuredExample} are obtained from a simulation of the corresponding non-structured model~\eqref{equ:non-structuredExample} using parameters according to~\ref{equ:paramSet}. They present the total cancer cell density~$C(t,x)$, ECM density~$v(t,x)$, and the free molecular species volume concentrations~$\vect m(t,x) = (m_1(t,x),m_2(t,x))^\T$ in the spatial domain~$\D$ at initial time $t = 0$ and at times $t = 50,100,150$, and $200$ (from left to right).

\begin{figure}
\includegraphics[width=\textwidth]{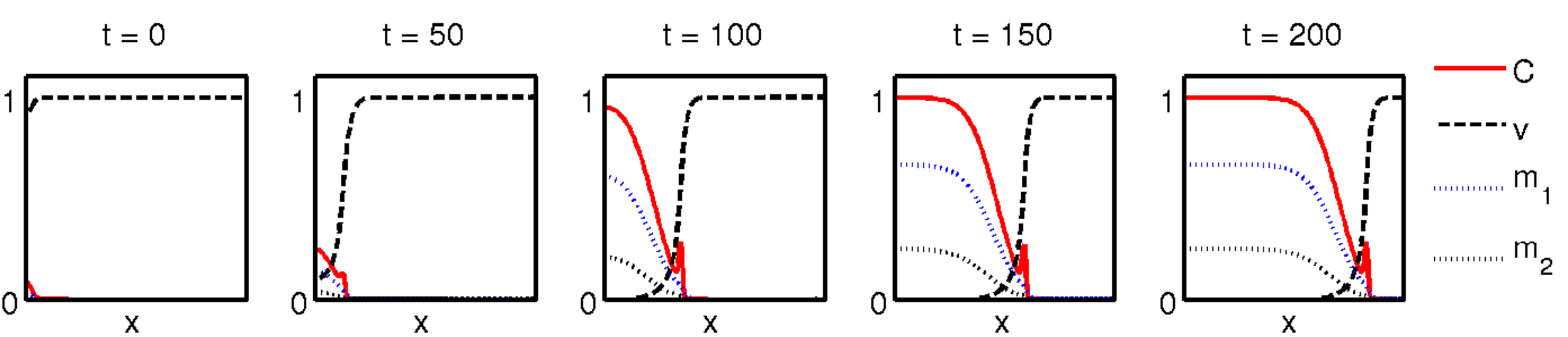}
  \caption{Plots showing the computational simulation results  at increasing time points (left to right) of the corresponding non-structured system~\eqref{equ:non-structuredExample} using the basic parameter set~\ref{equ:paramSet}.}\label{fig:nonStructuredExample}
\end{figure}
%
\begin{figure}
\begin{subfigure}[b]{\textwidth}
\includegraphics[width=\textwidth]{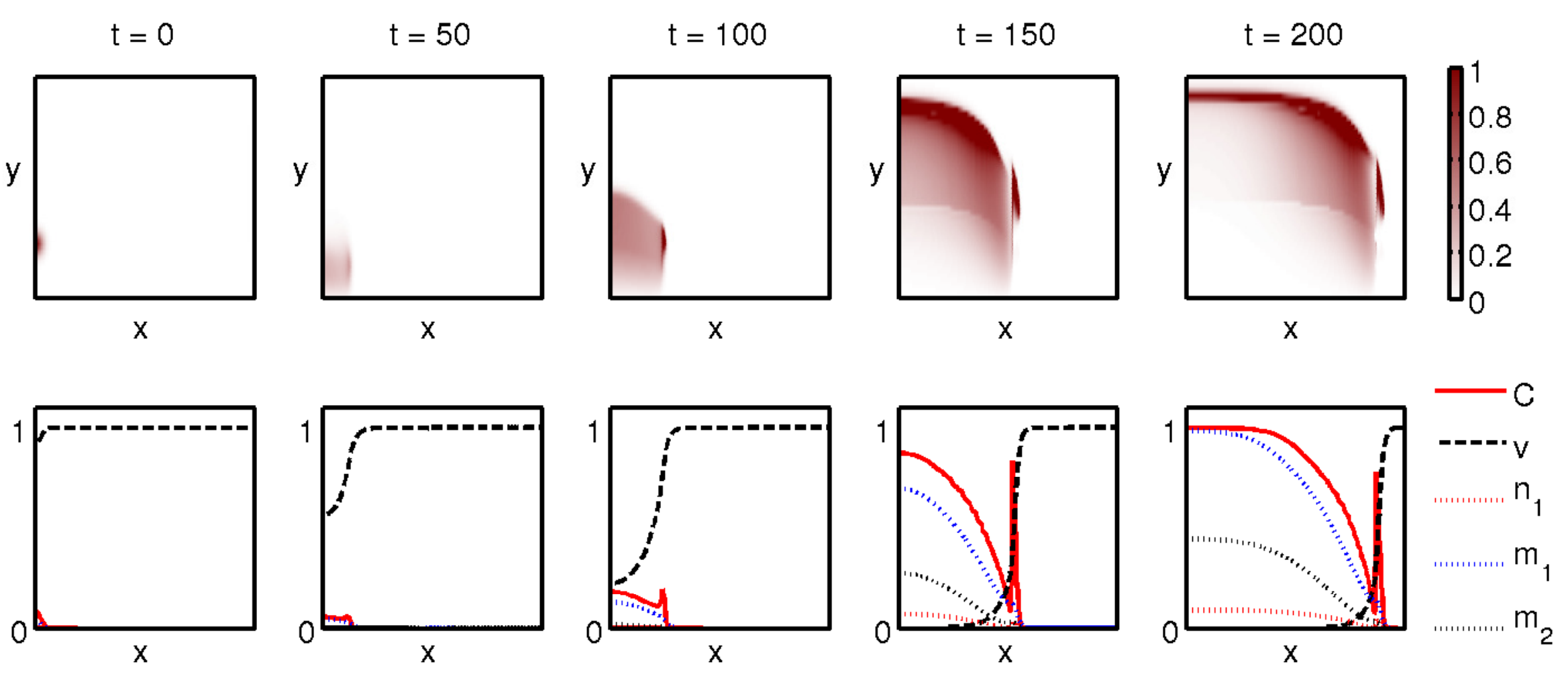}
\caption{
constant $\chi_v$}
\label{fig:structuredExample-withUnbinding-constantHapto}
\end{subfigure}\\
\begin{subfigure}[b]{\textwidth}
\includegraphics[width=\textwidth]{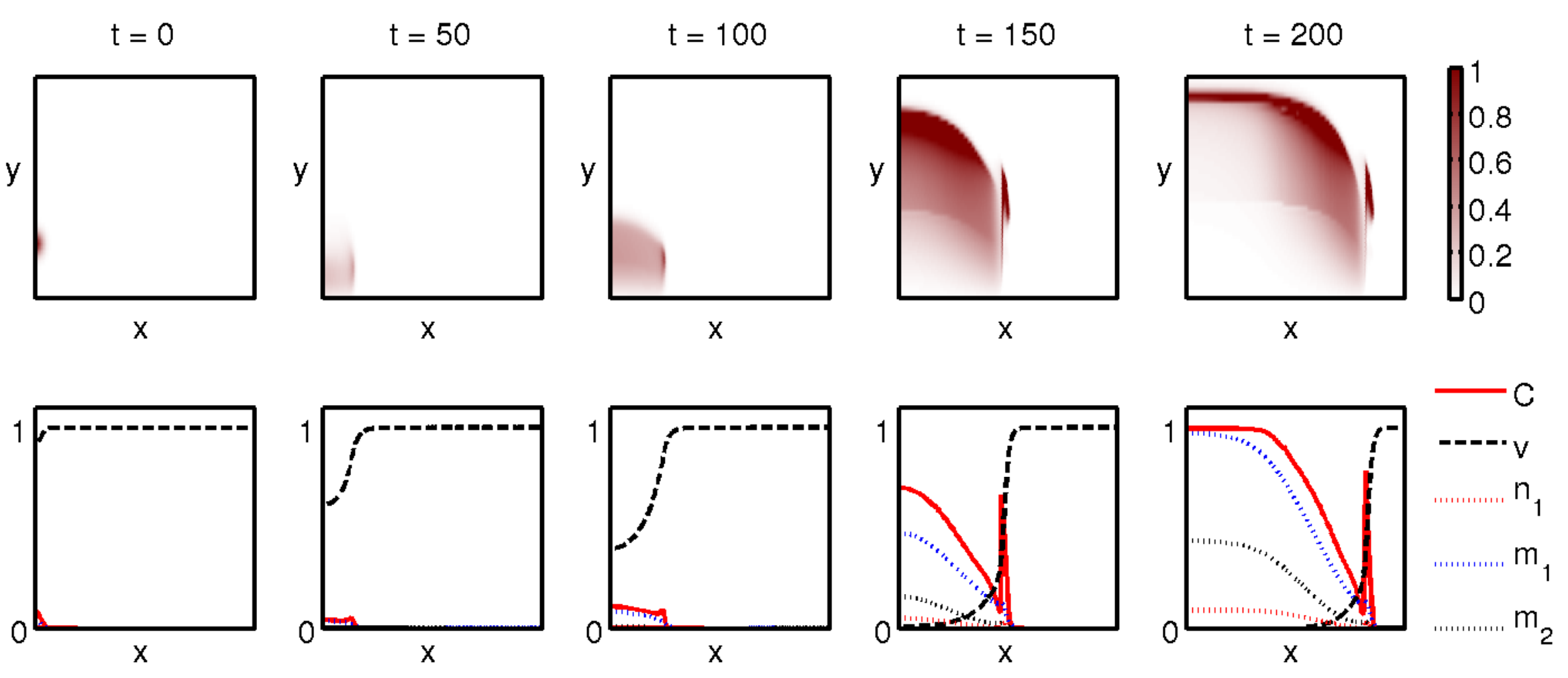}
\caption{
\istate-dependent $\chi_v$}
\label{fig:structuredExample-withUnbinding-linearHapto}
\end{subfigure}
\caption{Plots showing the computational simulation results at increasing time points (left to right) of the structured system~\eqref{equ:structuredExample} using the basic parameter set~\ref{equ:paramSet} with unbinding of molecules at the rate $\delta_y=0.05$ and using the haptotaxis term $\chi_v$ as specified in (a) and (b).
The top row in each of (a) and (b) shows the evolution of the structured cancer cell density in the spatio-structural space $\D\times\P$; the bottom row in each of (a) and (b) shows the evolution of all non-structured variables in the spatial domain~$\D$.}
\label{fig:structuredExample-withUnbinding}
\end{figure}

In Figure~\ref{fig:structuredExample-constantHapto}, we see that initially only a small amount of activator is bound to the cell surface and hence, up to $t=100$, the ECM is degraded much more slowly than in the non-structured case in Figure~\ref{fig:nonStructuredExample}. Over time, the cancer cells proliferate, produce, and bind more of the $m_1$-molecules, which in turn activate the matrix-degrading enzyme $m_2$. Hence, at later times, the level of MDEs $m_2$ is about twice as much in the structured case as in the non-structured case shown in Figure~\ref{fig:nonStructuredExample}. 

If we compare the constant haptotaxis result in Figure~\ref{fig:structuredExample-constantHapto} with the \istate-dependent haptotaxis case from Figure~\ref{fig:structuredExample-linearHapto}, we observe that the invading front of the total cancer cell density~$C(t,x)$ has a steeper and less regular shape. At the same time, a comparison between the spatio-structural dynamics shown in the top rows of Figures~\ref{fig:structuredExample-constantHapto} and~\ref{fig:structuredExample-linearHapto} reveals that until $t=150$ the amount of bound $m_1$-molecules, i.e. $n_1$, is less in the \istate-dependent case, leading to a slower start in ECM degradation while the cancer cells $c(t,x,y)$ remain less spread in the structural variable.

In Figure~\ref{fig:structuredExample-withUnbinding} we present the simulation results of the model~\eqref{equ:structuredExample}, where we consider the unbinding of molecules with rate $\delta_y=0.05$. We observe that due to the unbinding of the $m_1$-molecules, the degradation of the ECM is less compared to the case without unbinding, and the cell-surface concentration remains below the maximum of 1, i.e., $y<1$. 
A stronger aggregating tendency in the \istate component of the spatio-structural distribution of the invading cancer cells could be observed in Figure~\ref{fig:structuredExample-withUnbinding}, with the leading peak being higher compared to that in Figure~\ref{fig:structuredExample}. 




\section{A Structured-Population Model of Cancer Invasion Based on the uPA System}\label{sec:uPASystem}

After exploring the structured-population approach for the generic model of cancer invasion, we now apply the general framework to a more involved model of cancer invasion. We will also present the corresponding non-structured model and compare it to an already existing model for the same process.

Cancer cell invasion is a complex process occurring across many scales, both spatial and temporal, ranging from  biochemical intracellular interactions to cellular and tissue scale processes. A major component of the invasive process is the degradation of the extracellular matrix (ECM) by proteolytic enzymes. One important enzymatic system in cancer invasion that has been investigated in the literature is the so-called uPA system (urokinase plasminogen activation system), see for example \citet{Chaplain2005,Chaplain2006,Andasari2011}. 
It consists of a cancer cell population, the ECM, urokinase plasminogen activator (uPA) alongside plasminogen activator inhibitor type-1 (\PAI1) proteins, and the matrix degrading enzyme plasmin. These are accompanied by   urokinase plasminogen activator receptor (uPAR) molecules that are located on the cancer cell membrane.  

The free uPA molecules bind to uPAR and this complex subsequently activates the matrix degrading enzyme plasmin from its pro-enzyme plasminogen. In healthy cells, the activation of plasminogen is tightly regulated by the availability of uPA, for example by producing inhibitors of uPA like \PAI1. In contrast, cancer cells produce uPA to activate plasminogen, and hence excessively degrade the ECM, this way making room for further invasion. A schematic diagram can be found in Figure~\ref{fig:cancerCellSurface}. Details about the uPA system from a biological point of view can be found for example in \citet{Andreasen1997,Duffy2004,Ulisse2009}.

\begin{figure}[h]
\includegraphics[width=\textwidth]{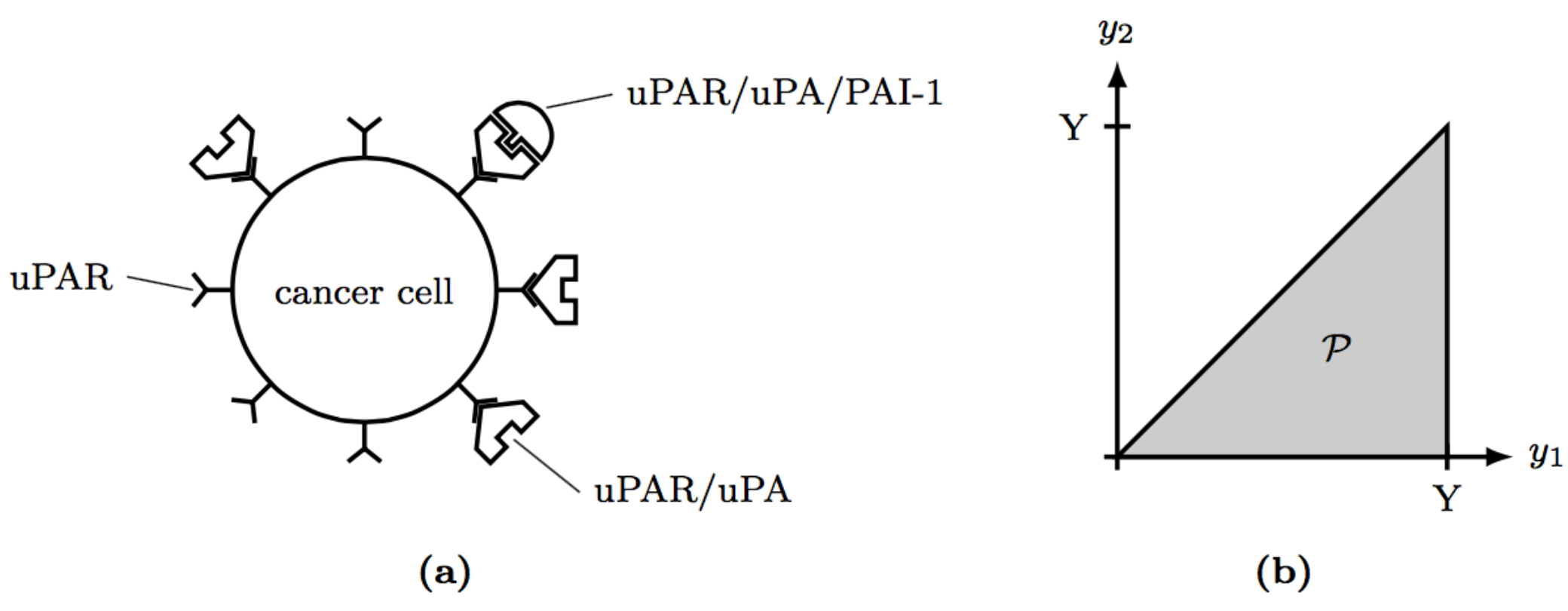}
\caption{Schematic diagrams of (a) a cancer cell with surface-bound receptors uPAR, bound uPA and inhibitor \PAI1; (b) the corresponding \istate space~$\P$.}
\label{fig:iStateSpace}
\end{figure}

Our structured general modeling framework \eqref{equ:generalModel} specialises for the uPA system using the following dependent variables:
\begin{itemize}
   \item the structured cancer cell density $c(t,x,y)$;
   \item the extracellular matrix density $v(t,x)$;
   \item the free molecular species volume concentrations, written as $$\vect m = (m_1, m_2, m_3)^\T\,,$$ where $m_1(t,x)$ represents the uPA, $m_2(t,x)$ stands for the \PAI1, and $m_3(t,x)$ is the plasmin volume concentration.
\end{itemize}

Here we assume that cancer cells carry a fixed amount 
of uPAR bound to their surface, hence the binding of uPA to the surface is limited by a maximal surface concentration $Y$.
Free \PAI1 enzymes only bind to the bound uPA. 
The two-dimensional \istate $y=(y_1,y_2)^\T \in \P$ therefore consists of the surface concentration $y_1$ of bound uPA on the cell surface, and the surface concentration $y_2\leq y_1$ of bound inhibitor \PAI1 molecules attached to the bound uPA enzymes.
Hence, the \istate space $\P$ is given by the open triangle $\P = \{y \in (0,Y)^2: y_2 < y_1\}$, as illustrated in the schematic diagram shown in Figure~\ref{fig:iStateSpace}.

For the binding rate of uPA, $b_1$, we assume that it is proportional to the free (unoccupied) receptors $Y-y_1$ and also to the availability of the free uPA, $m_1$. The binding rate of the inhibitor \PAI1, $b_2$, is assumed to be proportional to the uninhibited bound uPA $y_1-y_2$ and the availability of free \PAI1. Similarly, we assume the unbinding rate $d_2$ of \PAI1 to be proportional to the bound \PAI1, i.e., $y_2$. 
In this model, we do not consider that a uPA/\PAI1 complex unbinds as a whole but that first the \PAI1 must unbind.
Hence, the unbinding rate $d_1$ of uPA is proportional to the bound but uninhibited uPA, i.e. $y_1-y_2$. This gives rise to the following rates
\begin{align}\label{equ:uPA-bindingUnbinding}
 \vect{b}(y,\vect{m}) &= \begin{pmatrix}
		      (Y-y_1)\beta_1 m_1\\
		      (y_1-y_2) \beta_2 m_2
		  \end{pmatrix}\,, \qquad\text{and} &
 \vect{d}(y) &= \begin{pmatrix}
		      (y_1-y_2)\delta_{y_1}\\
		       y_2 \delta_{y_2}
		  \end{pmatrix}\,.
\end{align}

While the uPA is produced by the cancer cells, and the inhibitor \PAI1 is produced via plasmin activation, plasmin itself is activated from plasminogen by uninhibited bound uPA, which is described by $n_1-n_2$.
Hence, with~$\vect u$ and~$\vect r$ defined in~\eqref{equ:def_u_r}, we obtain that the vector of linear production terms is given by
\begin{align*}
	\boldsymbol \psi_{\vect m}(\vect u,\vect r) &= 
			\begin{pmatrix}
				\alpha_{m_1} C\\[1mm]
				\alpha_{m_2} m_3\\[1mm]
				\alpha_{m_3} (n_1-n_2)
			\end{pmatrix}\,.
\end{align*}

Finally, using the \istate-independent logistic proliferation law $\Phi (C,v)$ defined in~\eqref{equ:proliferationRate}, we arrive at the following system
\begin{subequations}
\begin{align}
   \begin{split}
     \frac{\dd c}{\dd t} &= \nabla_x \!\cdot\! \Bigl( D_c \nabla_x c - c(1-\rho(C,v)) \left(\chi_1\! \nabla_x m_1 +\chi_2 \!\nabla_x m_2 + \chi_v \!\nabla_x v\right)\Bigr)\\
      & \quad -\nabla_y \!\cdot\! \Bigl((\vect b(y,\vect m)-\vect d(y))c\Bigr) 
    + \Phi(C,v) \left(8c(t,x,2y) - c(t,x,y)\right)\,,
   \end{split} \label{equ:uPASystem-c}\\[2mm]
   \frac{\dd v}{\dd t} &= -\delta_v m_3 v + \mu_v (1-\rho(C,v))^+\,, \label{equ:uPASystem-v} \\[2mm]
   \begin{split}
   \frac{\dd m_1}{\dd t} &= \nabla_x \!\cdot\! \left[D_{m_1}\! \nabla_x  m_1\right] 
            \!-\! \bigl((Y\varepsilon C \!-\! n_1)\beta_1 m_1 \!-\! (n_1\!-\!n_2)\delta_{y_1} \bigr)\\
            &\qquad+ \alpha_{m_1} C -\delta_{m_1}m_1\,, \label{equ:uPASystem-m1}
   \end{split}
            \\[2mm]
     \frac{\dd m_2}{\dd t} &= \nabla_x \!\cdot\! \left[D_{m_2} \!\nabla_x m_2\right] 
             \!-\! \bigl( (n_1 \!-\! n_2)\beta_2 m_2 \!-\! n_2\delta_{y_2} \bigr)
            \!+ \alpha_{m_2} m_3 -\delta_{m_2}m_2\,, \label{equ:uPASystem-m2}
     \\[2mm]
      \frac{\dd m_3}{\dd t} &= \nabla_x \!\cdot\! \left[D_{m_3} \nabla_x  m_3\right] 
             + \alpha_{m_3}(n_1-n_2) - \delta_{m_3} m_3 \,. \label{equ:uPASystem-m3}
\end{align} \label{equ:uPASystem}
\end{subequations}

As before, system~\eqref{equ:uPASystem} is supposed to hold for $t\in\I$, $x\in\D$ and $y\in\P$ and is completed by initial conditions and appropriate boundary conditions. In space we have again zero-flux boundary conditions, while we have to determine the inflow boundary $\partial\P_{in}(t,x)$ as defined in~\eqref{equ:inflowBoundary} for the \istate space.
Here, the \istate space is defined as a triangle, see Fig.~\ref{fig:iStateSpace}, and we can divide the boundary into three parts, $\partial\P = \overline{\partial\P_1\cup\partial\P_2\cup\partial\P_3}$ with $\partial\P_1 :=\{(y_1,0):0<y_1<Y\}$, $\partial\P_2 :=\{(Y,y_2):0<y_2<Y\}$, and $\partial\P_3 :=\{(y_1,y_1):0<y_1<Y\}$, 
and the corresponding outer unit normal vectors
\[
 \normal(y) = \begin{cases}
               (0,-1)^\T\,, & \text{for } y\in\partial\P_1\,,\\
               (1,0)^\T\,, & \text{for } y\in\partial\P_2\,,\\
               \tfrac{1}{\sqrt{2}}(-1,1)^\T\,, & \text{for } y\in\partial\P_3\,.
              \end{cases}
\]
Then we obtain 
\[
 (\vect{b}(y,\vect{m}-\vect d(y))\cdot \normal(y) = \begin{cases}
               -y_1 \beta_2 m_2 \leq0\,, & \text{for } y\in\partial\P_1\,,\\
               -(Y-y_2)\delta_{y_1} \leq0\,, & \text{for } y\in\partial\P_2\,,\\
               -\tfrac{1}{\sqrt{2}}\left((Y-y_1)\beta_1m_1+y_1\delta_{y_2}\right) \leq0\,, & \text{for } y\in\partial\P_3\,.
              \end{cases}
\]
Provided that there exist molecules of the first and second type (meaning that the volume concentrations $m_1(t,x)$ and $m_2(t,x)$ are positive) and 
that binding and unbinding may take place (meaning that the binding and unbinding rate parameters $\beta_{1/2}$ and $\delta_{y_{1/2}}$ are positive),
all terms above are negative and the inflow boundary is $\partial\P_{in}(t,x)=\partial\P$. 
In case $m_1(t,x)$ or $m_2(t,x)$ or some of the parameters are zero, we might not have a classical inflow boundary at some parts of~$\partial\P$. 

As derived in Section~\ref{sec:generalDerivationNonStructuredModel} and using the mean values of the \istate-dependent coefficients, 
\begin{align*}
 \bar y &= \begin{pmatrix}\frac{2}{3}Y\\[1mm] \frac13 Y\end{pmatrix}\,,  &
 \bar{\vect{b}}(\vect m) &=  \begin{pmatrix}(Y-\bar y_1)\beta_1 m_1 \\[1mm] (\bar y_1-\bar y_2)\beta_2 m_2\end{pmatrix} \,, &
 \bar{\vect{d}} &= \begin{pmatrix}(\bar y_1 - \bar y_2) \delta_{y_1}\\\bar y_2 \delta_{y_2}\end{pmatrix}\,, &
\end{align*}
we obtain the following corresponding non-structured model for the dynamics of the uPA system

\begin{subequations}
\begin{align}
   \begin{split}
     \frac{\dd C}{\dd t} &= \nabla_x\!\cdot\!\! \Bigl( \bar D_c \nabla_x C - C(1\!-\!\rho(C,v))\! \left(\bar \chi_1 \nabla_x m_1 +\bar \chi_2 \nabla_x m_2 + \bar \chi_v \nabla_x v\right)\!\Bigr) \\
     &\qquad+ \Phi(C,v) C\,,
   \end{split} \label{equ:uPASystem-unstructured-c} \\[2mm]
   \frac{\dd v}{\dd t} &= -\delta_v m_3 v + \mu_v (1-\rho(C,v))^+\,,  \label{equ:uPASystem-unstructured-v} \\[2mm]
   \begin{split}
    \frac{\dd m_1}{\dd t} &= \nabla_x \!\cdot\! \left[D_{m_1} \!\nabla_x  m_1\right] - 
            \bigl((Y\!-\bar y_1)\beta_1 m_1 \!- (\bar y_1 - \bar y_2) \delta_{y_1} \bigr)\!\varepsilon C \\ 
            &\qquad + \alpha_{m_1} C \!- \delta_{m_1} m_1\,,  
   \end{split}
    \label{equ:uPASystem-unstructured-m1} \\[2mm]
     \frac{\dd m_2}{\dd t} &= \nabla_x \!\cdot\! \left[D_{m_2} \!\nabla_x m_2\right] \!-\! 
            \left((\bar y_1\!\!-\!\bar y_2)\beta_2 m_2 \!-\! \bar y_2 \delta_{y_2} \right)\!\varepsilon C \!+ \alpha_{m_2} m_3 \!-\! \delta_{m_2} m_2\,, \label{equ:uPASystem-unstructured-m2}
     \\[2mm]
      \frac{\dd m_3}{\dd t} &= \nabla_x \!\cdot\! \left[D_{m_3} \nabla_x  m_3\right] 
             + (\bar y_1 - \bar y_2) \alpha_{m_3} \varepsilon C - \delta_{m_3} m_3 \,. \label{equ:uPASystem-unstructured-m3}
\end{align} \label{equ:uPASystem-unstructured}
\end{subequations}

The unstructured system \eqref{equ:uPASystem-unstructured} obtained this way is similar in flavour to the one initially proposed by \citet{Chaplain2005,Chaplain2006}. The first differences appear though in equation~\eqref{equ:uPASystem-unstructured-v} and are due to the fact that our general structured framework~\eqref{equ:generalModel} assumed the simplified scenario for the ECM concentration evolution that is based only on enzymatic degradation and volume filling remodelling. In their special model  \citep{Chaplain2005,Chaplain2006}, the binding and unbinding of the \PAI1 inhibitor to and from the ECM as well as to and from the free uPA is taken into account, as well. These aspects show up also in the subsequent equations of the model proposed in \citet{Chaplain2005,Chaplain2006} concerning the dynamics of uPA, \PAI1, and plasmin, which cause them to differ in this regard from \eqref{equ:uPASystem-unstructured-m1}-\eqref{equ:uPASystem-unstructured-m3}. 

On the other hand, while in \eqref{equ:uPASystem-unstructured-m2} the process of \PAI1 inhibitor $m_2$ leaving the system through binding to the surface-bound uPA is captured by the structured framework \eqref{equ:uPASystem-m2}, in the corresponding equation from \citet{Chaplain2005,Chaplain2006} this is modelled by having the \PAI1 binding to the free uPA. Also, while \citet{Chaplain2005,Chaplain2006} assume a co-localisation of uPA and uPAR to activate plasmin, in our structured case \eqref{equ:uPASystem-m3}, the plasmin $m_3$ is explicitly activated by uninhibited bound uPA $n_1-n_2$, this leading to the non-structured approximation~\eqref{equ:uPASystem-unstructured-m3} expressed by a quantitatively derived proportionality to the cell surface distribution $\varepsilon C$. Future work will explore further similarities and discrepancies between the proposed structured and non-structured uPA models in an integrated computational and analytical approach.




\section{Conclusions and Outlook}\label{sec:conclusion}

In this paper we have established a general \emph{spatio-temporal-structural} framework that allows to describe the interaction of cell population dynamics (i.e. cell movement and proliferation) with molecular binding processes.
Any such structured model is complemented with a corresponding non-struc\-tured, spatio-temporal model. The latter is obtained by integrating the structured model over its \istate space.
Two specific examples, motivated by the process of cancer invasion, illustrate the applicability of the general structured framework and highlight the differences to the corresponding non-structured models.

In the first example, a generic model of cancer invasion, we observe numerically that the overall dynamics of the structured model differs in some regard from the corresponding non-structured one. This finds expression, for example, in a slower or faster degradation of the ECM depending on the amount of bound molecules, a different shape, speed, and intensity of the invading front, or different levels of the free (matrix-degrading) molecules.

In the second example, a model for the uPA-system, we compare the corresponding non-structured model with an existing non-structured model from \citet{Chaplain2005,Chaplain2006}. Our structured model is a more faithful representation of the underlying biology and structural information is inherited by the corresponding non-structured model and may lead to different terms compared to the existing non-structured model. This is evident, for example, in the term modelling the activation of plasmin, which is, as described in the biological literature, activated by cell-membrane bound but uninhibited uPA. While the model from the literature assumes activation via co-localisation of uPA and cancer cells (i.e. uPAR) but does not directly account for the binding to the cell membrane, our non-structured model uses the \istate mean value of the uninhibited bound uPA and thus incorporates, in a condensed form, structural information. Also, while in the existing non-structured model free uPA and \PAI1 are removed 
from the system upon contact as free uPA/\PAI1 complexes, they bind to the cell membrane and accordingly reduce the free uPA and \PAI1 volume concentration in our case; the internalisation of the uPAR/uPA/\PAI1 complex by the cell is discussed further below.

The benefit of this general structured model is that complex biological processes like binding to or unbinding from the cell's surface can be modelled quite naturally. We are able to distinguish between free and bound molecules, which can induce different reaction processes as was motivated biologically by the uPA system. Further, the bound molecules implicitly move with the cells, while the free molecules follow their own brownian motion. Moreover, the corresponding non-structured model, being an approximation of the structured one, inherits some of the structural information. Although the structured ansatz is computationally more expensive due to the additional dimensions of the \istate space, it allows a more realistic modelling of the underlying biological processes. 

The derivation of the general structured model~\eqref{equ:generalModel} as well as the corresponding non-structured model~\eqref{equ:generalNonStructuredModel} is based on a number of assumptions and simplifications. 
We comment on a selection of these in some detail below but leave their thorough discussion for follow-up work.

\paragraph{Spatial flux generalizations.} In the general model~\eqref{equ:generalModel}, we use, for the structured cell density $c$, the spatial flux term~\eqref{equ:spatialFlux}, which consists of a combination of diffusion, chemotaxis, and haptotaxis.

The diffusive flux term in~\eqref{equ:spatialFlux} is chosen as $-D_c\nabla_x c$. The same form is used, for instance, in the work of \citet{Laurencot2008}, who consider an age-structured spatio-temporal model for \emph{proteus mirabilis} swarm-colony development. This form implies that the random motility of cells with a particular \istate $y$ depends only on the gradient of the density of cells having that same \istate.
Instead, one could also think of random motility of cells at a particular \istate $y$ which is governed by the gradient of the total cell density. This would lead to a diffusive flux term of the form $-D_c \nabla_x C$. 

A further generalization of the spatial flux term is to consider cell movement due to cell-cell and cell-matrix adhesive interactions, as is done in a non-structured situation in \citet{Armstrong2006} and \citet{Gerisch2008}. The formulation of the required, so-called adhesion velocity~$\Ard$ will then have to be extended to the structured case and could be defined as
\begin{align*}
   \Ard(t,x,y,\vect u(t,\cdot)) =
   \frac{1}{R}\int\limits_{B(0,R)} \normal(\tilde x)
  \Omega(\nor{\tilde x}_2) g(t,y,\vect u(t,x+\tilde x))\dif{}\tilde x
\end{align*}
with the sensing radius $R>0$, $\normal(\tilde x)$ a unit normal vector pointing from $x$ to $x + \tilde{x}$, and the radial dependency function $\Omega(r)$. 
Similar as in the discussion for the diffusive flux above, cell adhesion occurs not only between cells of the same \istate but between cells of all \istate{}s. Assuming the adhesive strength to be identical for cells of all \istate{}s, the adhesion coefficient function $g$ will have the form
\[
 g(t,\vect{u}) \equiv \vect{g}(t,C,v) = \left[S_{cc}(t)C + S_{cv}(t)v\right] \cdot
  \left(1-\rho(C,v)\right)^+\,,
\]
where we have that $S_{cc}(t)$ represents the cell-cell adhesion coefficient, and $S_{cv}(t)$ denotes the cell-matrix adhesion coefficient.
This \istate-independent adhesion coefficient function coincides with the original one from \citet{Armstrong2006} and \citet{Gerisch2008}, and the time-dependent extension as studied in \citet{Domschke2014}.
In the structured case, cell-cell and cell-matrix adhesion can be influenced by the \istate of the cells, hence the adhesion coefficients would also depend on the \istate{(s)}. 
In order to take all \istate{}s into account, we have to integrate the cell-cell adhesion term over the \istate space. The adhesion coefficient function will then have the form
\begin{multline*}
   g(t,y,\vect u(t,x^*)) = \\
   \left( \int\limits_\P \!\! S_{cc}(t,y,\tilde y) c(t,x^*,\tilde y) \dif{}\tilde y  + S_{cv}(t,y) v(t,x^*)\right) \!\!\cdot\! \Bigl(\!1 \!-\! \rho\bigl(C(t,x^*),v(t,x^*)\bigr)\!\Bigr)^+,
\end{multline*}
where~$S_{cc}(t,y,\tilde y)$ represents the cell-cell adhesion coefficient between cells of \istate{s}~$y$ and~$\tilde y$, respectively. $S_{cv}(t,y)$ denotes the cell-matrix adhesion coefficient of cells with \istate~$y$ and the ECM. 
Both extensions of the spatial flux term need to be analysed in more detail and are subject to further investigation.

\paragraph{Internalisation.} In the general model~\eqref{equ:generalModel} we describe how binding and unbinding of the molecules influence the dynamics of the overall system. Free molecules leave and enter the system due to binding and unbinding, while the structured cell population ``moves'' through the \istate space.
In \citet{Cubellis1990}, it is described that surface-bound uPA/uPAR complexes are internalised and degraded by the cells. To include this mechanism in our model, we would have to add an internalisation term, similar to the unbinding term, to the structural flux~\eqref{equ:structuralFlux}. Since the uPA/uPAR complexes are degraded, they would not reenter the system as they do in the case of unbinding, hence the internalisation term would not appear in the free molecular species equation~\eqref{equ:generalModel-motile}.

\paragraph{Variable receptor density.} In the special case of the uPA system, we assume, following the work of \citet{Chaplain2005,Chaplain2006}, that a cancer cell carries a fixed amount of uPAR on its cell surface. However, one could assume a varying surface density of uPAR due to external influence or active alteration by the cancer cells. \citet{Yang2006}, for example, have shown that subpopulations of colon cancer cells with an initially low cell surface uPAR number can spontaneously develop an oscillating cell surface uPAR density. In our general modelling framework it is possible to capture such a mechanism by adding an additional \istate variable describing the surface concentration of uPAR.

\paragraph{Intracellular reactions.} In this work, we describe how to model surface-bound reactions in a structured-population approach. 
Such a structured approach is also suitable to describe intracellular reactions and the effect of the exchange of molecules between the cell's cytoplasm and the extracellular space or even the cell membrane. 
The corresponding changes in cell state will in many cases have an influence on the cell's behaviour.
These processes can be expressed by making use of a structured cell volume density which is defined assuming a fixed volume for each cell. 
The latter is analogous to the structured cell surface density $s(t,x,y)$, as considered in this work, for which we assume a fixed cell surface area $\varepsilon$.

\paragraph{Higher order approximations in non-structured model.} In the derivation of a non-structured model from the general structured model~\eqref{equ:generalModel} it is necessary to ap\-prox\-i\-mate terms involving structured expressions with expressions that involve only non-structured variables. In Section~\ref{sec:generalDerivationNonStructuredModel}, we use \istate mean values of all corresponding structured terms. Basically it is possible to consider more sophisticated, higher order approximations of these terms in order to incorporate the structural information in a more refined manner.

\paragraph{Structural flux and relation to age-structured models.} Our general model~\eqref{equ:generalModel-cells} as well as age-structured models are hyperbolic in the \istate/age variable.
Accordingly, the prescription of boundary conditions on the \istate space has to be handled with care and is only possible at inflow boundary parts.
The transport coefficient w.r.t. age in age-structured models is constant and uniform and thus the inflow boundary is \emph{a priori} known and no ``crossing of characteristics'' is possible.
In contrast to that, the transport coefficient in our structured model is given by the net binding rate, which depends on the \istate $y$ and the free molecular volume concentrations $\vect{m}$. 
It is thus in general a nonuniform and nonlinear expression and hence changes in the \istate and with time. Thus, firstly, the inflow boundary, where the scalar product of the net binding rate and the unit outward normal vector is negative, may change with time and also a ``crossing of characteristics'' is possible. Further analytical investigations are required to give more insight into these issues and, more general, addressing rigorously the existence, uniqueness, and positivity of solutions of the proposed model.


\appendix

\section{A Measure Theoretic Setting} \label{app:bindingMeasures}

A measure theoretical justification of the binding and unbinding rates introduced to define the structural flux given in \eqref{equ:structuralFlux} is as follows. 
Let $\borel(\P)$ denote the Borel $\sigma-$algebra of the \istate space~$\P$.
In our model, given a density of molecular species~$\vect m(t,x)$, the structural measure of their binding rate to the total cell density~$C(t,x)$ is denoted by $\eta_{\vect b}(\cdot;\vect m):\borel(\P)\to \R^p$ and is assumed to be absolutely continuous with respect to the Lebesgue measure on~$\P$. Then the induced Lebesgue-Radon-Nikodym density 
\begin{align}
   \vec{b}(\cdot;\vect{m}) = \begin{pmatrix}b_1(\cdot;\vect m)\\ \vdots \\ b_p(\cdot;\vect m)\end{pmatrix} : \P \to \R^p\,.
   \label{equ:bindingRate}
\end{align}
is uniquely defined by 
\begin{align}
	\eta_{\vect b} (W;\vect m) &= \int\limits_W \vect b(\gamma;\vect m) \dif \gamma\,, \quad \forall W\in\borel(\P)\,,
	\label{equ:defBindingRate}
\end{align}
\citep{Halmos1978}, and represents the binding rate of the molecular species~$\vect m$ to the cell population density~$c$. 

Similarly, the structural measure of their unbinding rate of the bound molecular species~$\vect n(t,x)$ is denoted by~$\eta_{\vect d}$ and is again assumed to be absolutely continuous with respect to the Lebesgue measure on~$\P$. Thus, this leads to an unbinding rate depending only on the \istate given by the  Lebesgue-Radon-Nikodym density 
\begin{align}
   \vect{d}(\cdot) = \begin{pmatrix}d_1(\cdot)\\ \vdots \\ d_p(\cdot)\end{pmatrix} : \P \to \R^p
   \label{equ:unbindingRate}
\end{align}
is uniquely defined by 
\begin{align}
	\eta_{\vect d} (W) &= \int\limits_W \vect d(\gamma) \dif \gamma\,, \quad \forall W\in\borel(\P)\,.
	\label{equ:defUnbindingRate}
\end{align}



\section{The Source Term for Arbitrary Borel Sets $W\subset \P$} \label{app:sourceTerm}

Let $W \subset \P \subset \R^p$ be an arbitrary Borel set and define $zW := \{zw: w\in W\} $ for $z\in\R$. If $z \neq 0$, we can also write $zW = \{\tilde w: \frac1z \tilde w \in W\}$. 
Assume that $W$, $2W$, and $\frac12 W$ are pairwise disjoint as shown in Fig.~\ref{fig:mitosis}. 
Then, the source of cells in the structural region $W$ that was obtained in \eqref{equ:sourceTermIntegral} reads as
\begin{align}
   \int\limits_W S(t,x,y) \dif y &= 2 \int\limits_{2W} \Phi(\tilde y,\vect u) c(t,x, \tilde y)\dif \tilde y - \int\limits_W \Phi(y,\vect u) c(t,x,y) \dif y\,. \label{equ:sourceTermIntegralAppendix}
\end{align}
Note that we may have to invoke Convention~\ref{conv:istate} in the evaluation of the integral over $2W$. 
The purpose of this appendix is to show that Eq.~\eqref{equ:sourceTermIntegralAppendix} also holds for arbitrary Borel sets $W\subset \P$.
We start with the following technical lemma.
\begin{lemma}
\label{lem:disjointSets}
  Consider a set $A$ such that $A\cap 2A = \emptyset$. Then it holds that
  $\frac12 A\cap A=\emptyset$ 
  and, provided that $A$ is also convex, 
  $\frac12 A\cap 2A=\emptyset$.
\end{lemma}
\begin{proof}
Suppose there exists an $x\in \frac{1}{2} A\cap A$. Then $x\in A$ and it exists a $y\in A$ such that $x=\frac12 y$. 
This implies that $y=2x$ is also an element of $2A$ and thus $y\in A\cap 2A$, a contradiction. 
Thus $\frac12 A\cap A=\emptyset$ must hold. 

Now suppose there exists an $x\in\frac{1}{2} A\cap 2A$. Then there exist $y,z\in A$ such that $x=\frac12 z$ and $x = 2y$.
Now observe that $x=\alpha y + (1-\alpha)z$ for $\alpha=\frac{2}{3}\in(0,1)$ and thus $x$ can be written as a convex combination of $y$ and $z$.
Since $A$ is convex, we also have $x\in A$. But then, $z=2x\in A\cap 2A$, a contradiction. 
Thus $\frac12 A\cap 2A=\emptyset$ must hold. 
\qed
\end{proof}
This enables us now to prove the following theorem. 

\begin{theorem}
Let $W$ be an arbitrary convex and compact subset of $\P$. Then Eq.~\eqref{equ:sourceTermIntegralAppendix} holds.
\end{theorem}

\begin{proof}
Since $W \subset \P$ is an arbitrary convex and compact set, the sets $W$, $2W$, and $\frac12 W$ might not be pairwise disjoint. Since $W$ is compact, we have that the Lebesgue measure $\lambda(W)<\infty$.
Furthermore, it holds that $\lambda(zW) = z^{p} \lambda(W)$ for all $z\in\R$.
We define the sequence of sets
\[
 W_k := \bigcap\limits_{i=0}^k 2^{-i} W, \quad k=0,1,2,\dots\,.
\]
These sets are, as intersection of convex sets, convex and have the following properties
\[
  W_0 = W\,,\quad
  W_j \subseteq W_i\text{ for all } j\geq i\,,\quad\text{and}\quad
  \lambda(W_k)\leq 2^{-pk} \lambda(W)\,.
\]
Note that if $0\not\in W$ then there exists a finite $K$ such that $W_k=\emptyset$ for all $k\ge K$, otherwise, if $0\in W$ then $0\in W_k$ for all $k$ and $\lim_{k\to\infty}W_k=\{0\}$.
Therefore, combining both cases, define $W_\infty:=\{0\}\cap W$; clearly $\lambda(W_\infty)=0$.
Thus we can write
\[
 W = W_0 = W_0\setminus W_1 \cupdot W_1 
 = \ldots 
 = \left(\bigcupdot_{i=0}^k W_i \setminus W_{i+1}\right)\cupdot W_{k+1}
 = \left(\bigcupdot_{i=0}^\infty A_i\right) \cupdot W_\infty\,,
\]
where $A_k:=W_k \setminus W_{k+1}$ for $k=0,1,\dots$.
From the definition of the sets $W_k$ we can also deduce the following relation
\[
  2W_k = 2W\cap W_{k-1}\text{ for } k=1,2,\dots\,.
\]
Now, for all $k=0,1,2,\dots$ we obtain
\begin{align*}
  A_k \cap 2 A_k
  &= (W_k \setminus W_{k+1}) \cap (2 W_k) \setminus (2 W_{k+1})\\
  &= (W_k \setminus W_{k+1}) \cap (2 W_k) \setminus (2W\cap W_{k})\\
  &= (W_k \setminus W_{k+1}) \cap \Bigl[\underbrace{((2 W_k)\setminus (2W))}_{=\emptyset} \cup\, ((2W_k)\setminus W_k) \Bigr]\\
  &= (W_k \setminus W_{k+1}) \cap ((2W_k)\setminus W_k)\\
  &= \emptyset\,.
\end{align*}
Following the first part of Lemma~\ref{lem:disjointSets} it now also follows that $\frac{1}{2} A_k \cap A_k =\emptyset$ for $k=0,1,2,\dots$.
The second part of Lemma~\ref{lem:disjointSets} is not applicable here since $A_k$ is in general not convex. 
However, note that the derivation above also shows that
\[
  2 A_k = (2W_k)\setminus W_k \text{ for } k=0,1,2,\dots\,.
\]
Using that relation, once directly and once multiplied by $\frac{1}{4}$, we now obtain, for all $k=0,1,2,\dots$,
\begin{align*}
  \frac{1}{2}A_k \cap 2 A_k
  &=\Biggl(\Bigl(\frac{1}{2}W_k\Bigr)\setminus \Bigl(\frac{1}{4}W_k\Bigr)\Biggr) \cap \bigl((2W_k)\setminus W_k\bigr)\,.
\end{align*}
Now assume that there exists an $x\in \frac{1}{2}A_k \cap 2 A_k$. Then necessarily, $x\in \frac{1}{2}W_k$ and $x\in 2 W_k$.
As in the proof of the second part of Lemma~\ref{lem:disjointSets} it follows, thanks to the convexity of $W_k$, that also $x\in W_k$.
However, then $x\not\in (2W_k)\setminus W_k$ and thus $x\not\in\frac{1}{2}A_k \cap 2 A_k$, a contradiction. 
Thus it also holds that $\frac{1}{2}A_k \cap 2 A_k=\emptyset$.

In summary, it holds that, for each $k=0,1,2,\dots$, the sets $A_k$, $2 A_k$, and $\frac{1}{2} A_k$ are pairwise disjoint and hence Eq.~\eqref{equ:sourceTermIntegralAppendix}
holds with $W$ replaced by $A_k$.

Now we can conclude for our arbitrary convex and compact set $W \subset \P$, that
\begin{align*}
   \int\limits_W S(t,x,y) \dif y 
   &= \int\limits_{\bigcupdot\limits_{i=0}^\infty A_i} S(t,x,y) \dif y = \sum\limits_{i=0}^\infty \int\limits_{A_i} S(t,x,y) \dif y \\
   &\stackrel{\eqref{equ:sourceTermIntegralAppendix}}{=} \sum\limits_{i=0}^\infty \left[ 2 \int\limits_{2A_i} \Phi(\tilde y,\vect u) c(t,x, \tilde y)\dif \tilde y - \int\limits_{A_i} \Phi(y,\vect u) c(t,x,y) \dif y \right]\\
   &= 2 \int\limits_{\bigcupdot\limits_{i=0}^\infty 2A_i} \Phi(\tilde y,\vect u) c(t,x, \tilde y)\dif \tilde y - \int\limits_{\bigcupdot\limits_{i=0}^\infty A_i} \Phi(y,\vect u) c(t,x,y) \dif y\\
   &= 2 \int\limits_{2W} \Phi(\tilde y,\vect u) c(t,x, \tilde y)\dif \tilde y - \int\limits_{W} \Phi(y,\vect u) c(t,x,y) \dif y\,.
\end{align*}
\qed
\end{proof}

Since we have shown that Eq.~\eqref{equ:sourceTermIntegralAppendix} holds for arbitrary convex and compact subsets of $\P$, it in particular also holds for all rectangles, which are a family of generators of the Borelian $\sigma$-algebra on $\P$ \citep{Halmos1978}. Hence it holds for all Borel subsets of $\P$.


\section{Non-Dimensionalisation and Parameter Tables}
\label{app:ParameterTables}

Based on a typical cancer cell volume of $1.5\times 10^{-8}\rm cm^3$, see \citet{Anderson2005} and references cited there, we set
\[
  \vartheta_c = 1.5\times 10^{-8}\rm{cm^3/cell}
\]
and define below the scaling parameter $c_\ast = 1/\vartheta_c = 6.7\times 10^{7}\rm{cells/cm^3}$ as the inverse of $\vartheta_c$, i.e., taken as the maximum cell density such that no overcrowding occurs. 
Assuming that a cell is approximately a sphere, we obtain a surface area of $\varepsilon = 2.94\times 10^{-5}\rm{cm^2/cell}$. 
In \citet{Lodish2007}, the amount of surface receptors is given by a range from 1,000 to 50,000 molecules per cell. We take the upper limit which is translated to $50,000$ $\rm{molecules/cell} = 8.3\times 10^{-14}\rm{\upmu mol/cell}$ and gives a reference surface density of 
\[
   y_\ast = \frac{8.3\times 10^{-14}\rm{\upmu mol/cell}}{2.94\times 10^{-5}\rm{cm^2/cell}}  = 2.82\times 10^{-9}\rm{\upmu mol/cm^2}. 
\]
In \citet{AbreuPalmerMurray2010} it is stated that the collagen density in engineered provisional scaffolds should be between $2$ and $4$~$\rm{mg/cm^3}${} for \emph{in vivo} delivery. We take the upper limit as scaling parameter $v_\ast$ for the ECM density. Assuming that ECM at this density fills up all available physical space, we obtain $1=\rho(\vect{0},v_\ast)=\vartheta_v v_\ast$ and thus
\[
\vartheta_v:=\frac{1}{v_\ast}\,.
\]
The scaling parameters $\tau=1\times 10^{4}\rm{s}$ and $L=0.1\rm{cm}$ are chosen as in \citet{Gerisch2008} and \citet{Domschke2014} and, as in loc.\,cit., the value of the scaling parameter $m_\ast$ remains unspecified. 
Table~\ref{tab:parameters} shows the model parameters with units and their non-dimensionalised counterparts, and intermediate quantities of these can be found in Table~\ref{tab:intermediateQuantities}.


\begin{table}[htb]
   \centering
\begin{tabular}{llll}
  \toprule
    $p$ & unit & $\tilde{p}$ & conditions\\
  \midrule
  \\[-5pt]
    $\varepsilon$ & $\rm{cm^2/cell}$& $\frac{c_* y_*}{m_*}\varepsilon$ & $\varepsilon >0$
    \\[5pt]
    $\vartheta_c$ & $\rm{cm^3/cell}$ &$c_\ast \vartheta_c$ &$\vartheta_c>0$
    \\[5pt]
    $\vartheta_v$ & $\rm{cm^3/mg}$ &$v_\ast \vartheta_v$  & $\vartheta_v>0$
    \\[5pt]
    $D_c$ & $\rm{cm^2/s}$  & $\displaystyle \frac{\tau}{L^2} D_c$ &  $D_c>0$
    \\[8pt]
    $\chi_k$ & $\rm{(cm^2/s)/n M}$ &  $\displaystyle \frac{\tau}{L^2}m_* \chi_k$ & $\chi_k\geq0$, $k=1,\dots,q$
    \\[8pt]
    $\chi_v$ & $\rm{(cm^2/s)/(mg/cm^3)}$  &  $\displaystyle \frac{\tau}{L^2}v_* \chi_v$ & $\chi_v\geq0$
    \\[8pt]
    $\boldsymbol\delta_v$ & $\rm{1/(n M s)}$  & $\tau m_\ast \boldsymbol \delta_v$&
    $\boldsymbol \delta_v \geq 0$
    \\[5pt]
    $\mat D_{\vect m}$ & $\rm{cm^2/s}$ & $\displaystyle \frac{\tau}{L^2} \mat D_{\vect m}$ &
    $\mat D_{\vect m}>0$
    \\[8pt]
    $\boldsymbol \delta_{\vect m}$ & $\rm{1/s}$ & $\displaystyle \tau \boldsymbol\delta_{\vect m}$ &
    $\boldsymbol \delta_{\vect m} \ge 0$
    \\
  \bottomrule
\end{tabular}
\caption{\label{tab:parameters}
Parameters $p$ of the general model \eqref{equ:generalModel}  with their unit and their
non\hyp{}dimensionalised counterparts $\tilde{p}$.}
\end{table}

\begin{table}
   \centering
\begin{tabular}{lllp{5.5cm}}
  \toprule
    $p$ & unit & $\tilde{p}$ &  references/notes\\
  \midrule
  \\[-5pt]
    $\rho(C,v)$   & --- &  $\rho$ & volume fraction of occupied space
    \\[8pt]
    $\Phi(y,\vect{u})$ & $\rm{1/s}$ & $\tau\Phi $
    & \istate-dependent cell proliferation rate
    \\[5pt]
    $\vect b(y,\vect{m})$ & $\rm{(\upmu mol/cm^2)/s}$ & $\displaystyle\frac{\tau}{y_*}\vect b$
    & vector of binding rates of molecular species, $\vect b \geq 0$,
    \\[5pt]
    $\vect d(y)$ & $\rm{(\upmu mol/cm^2)/s}$ & $\displaystyle\frac{\tau}{y_*}\vect d$
    & vector of unbinding/detaching rates of molecular species, $\vect d \geq 0$,
    \\[5pt]
    $\psi_v(t,\vect{u})$ & $\rm{(mg/cm^3)/s}$ & $\displaystyle
    \frac{\tau}{v_\ast}\psi$
    & ECM remodelling law, $\psi_v\ge 0$ if \mbox{$v=0$}
    \\[8pt]
    $\boldsymbol\psi_{\vect m}(\vect u, \vect r)$ & $\rm{n M/s}$  & $\displaystyle
    \frac{\tau}{m_\ast} \boldsymbol\psi_{\vect m}$ 
    & vector of production terms for molecular species, $\boldsymbol\psi_{\vect m}\ge\mathbf{0}$
    \\
  \bottomrule
\end{tabular}
\caption{\label{tab:intermediateQuantities}
Intermediate model quantities $p$ of the general model \eqref{equ:generalModel} with their unit and their non-dimensionalised counterparts $\tilde{p}$.
The latter have to be read, for instance, as follows
$\tilde{\vect{b}}(\tilde{y},\tilde{\vect{m}}) = \frac{\tau}{y_*}\vect b(y,\vect{m})$.}
\end{table}



\begin{acknowledgements}
PD was supported by the Northern Research Partnership PECRE scheme and the Deutsche Forschungsgemeinschaft under the grant DO~1825/1-1. 
DT and AG would like to acknowledge Northern Research Partnership PECRE scheme. 
DT and MAJC gratefully acknowledge the support of the ERC Advanced Investigator Grant 227619, ``M5CGS - From Mutations to Metastases: Multiscale Mathematical Modelling of Cancer Growth and Spread''.
The authors PD, DT, AG, and MAJC would like to thank the Isaac Newton Institute for Mathematical Sciences for its hospitality during the programme ``\emph{Coupling Geometric PDEs with Physics for Cell Morphology, Motility and Pattern Formation}'' supported by EPSRC Grant Number EP/K032208/1. 
\end{acknowledgements}

\bibliographystyle{spbasic}      

\bibliography{DomschkeTrucuGerischChaplain_ArXiv_version}   

\begin{thebibliography}{111}
\providecommand{\natexlab}[1]{#1}
\providecommand{\url}[1]{{#1}}
\providecommand{\urlprefix}{URL }
\expandafter\ifx\csname urlstyle\endcsname\relax
  \providecommand{\doi}[1]{DOI~\discretionary{}{}{}#1}\else
  \providecommand{\doi}{DOI~\discretionary{}{}{}\begingroup
  \urlstyle{rm}\Url}\fi
\providecommand{\eprint}[2][]{\url{#2}}

\bibitem[{Abia et~al(2009)Abia, Angulo, L{\'o}pez-Marcos, and
  L{\'o}pez-Marcos}]{Abia2009}
Abia L, Angulo O, L{\'o}pez-Marcos J, L{\'o}pez-Marcos M (2009) Numerical
  schemes for a size-structured cell population model with equal fission.
  Mathematical and Computer Modelling 50(5--6):653 -- 664,
  \doi{10.1016/j.mcm.2009.05.023}

\bibitem[{Abreu et~al(2010)Abreu, Palmer, and Murray}]{AbreuPalmerMurray2010}
Abreu EL, Palmer MP, Murray MM (2010) Collagen density significantly affects
  the functional properties of an engineered provisional scaffold. J Biomed
  Mater Res Part A 93A(1):150--157, \doi{10.1002/jbm.a.32508}

\bibitem[{Ainseba and Anita(2001)}]{Ainseba2001}
Ainseba B, Anita S (2001) Local exact controllability of the age-dependent
  population dynamics with diffusion. Abstract and Applied Analysis
  6(6):357--368, \doi{10.1155/S108533750100063X}

\bibitem[{Ainseba and Langlais(2000)}]{Ainseba2000}
Ainseba B, Langlais M (2000) On a population dynamics control problem with age
  dependence and spatial structure. Journal of Mathematical Analysis and
  Applications 248(2):455 -- 474, \doi{10.1006/jmaa.2000.6921}

\bibitem[{Al-Omari and Gourley(2002)}]{Al-Omari2002}
Al-Omari J, Gourley S (2002) Monotone travelling fronts in an age-structured
  reaction-diffusion model of a single species. J Math Biol 45(4):294--312,
  \doi{10.1007/s002850200159}

\bibitem[{Allen(2009)}]{Allen2009}
Allen EJ (2009) Derivation of stochastic partial differential equations for
  size- and age-structured populations. Journal of Biological Dynamics
  3(1):73--86, \doi{10.1080/17513750802162754}

\bibitem[{Andasari et~al(2011)Andasari, Gerisch, Lolas, South, and
  Chaplain}]{Andasari2011}
Andasari V, Gerisch A, Lolas G, South AP, Chaplain MA (2011) Mathematical
  modeling of cancer cell invasion of tissue: biological insight from
  mathematical analysis and computational simulation. J Math Biol
  63(1):141--171, \doi{10.1007/s00285-010-0369-1}

\bibitem[{Anderson and Chaplain(1998)}]{Anderson1998}
Anderson A, Chaplain M (1998) Continuous and discrete mathematical models of
  tumor-induced angiogenesis. Bull Math Biol 60(5):857--899,
  \doi{10.1006/bulm.1998.0042}

\bibitem[{Anderson(2005)}]{Anderson2005}
Anderson ARA (2005) A hybrid mathematical model of solid tumour invasion: the
  importance of cell adhesion. IMA Math Med Biol 22(2):163--186,
  \doi{10.1093/imammb/dqi005}

\bibitem[{Anderson et~al(2000)Anderson, Chaplain, Newman, Steele, and
  Thompson}]{Anderson2000}
Anderson ARA, Chaplain MAJ, Newman EL, Steele RJC, Thompson AM (2000)
  Mathematical modelling of tumour invasion and metastasis. J Theor Med
  2(2):129--154, \doi{10.1080/10273660008833042}

\bibitem[{Andreasen et~al(1997)Andreasen, Kj{\o}ller, Christensen, and
  Duffy}]{Andreasen1997}
Andreasen PA, Kj{\o}ller L, Christensen L, Duffy MJ (1997) The urokinase-type
  plasminogen activator system in cancer metastasis: A review. Int J Cancer
  72(1):1--22,
  \doi{10.1002/(SICI)1097-0215(19970703)72:1<1::AID-IJC1>3.0.CO;2-Z}

\bibitem[{Andreasen et~al(2000)Andreasen, Egelund, and
  Petersen}]{Andreasen2000}
Andreasen PA, Egelund R, Petersen HH (2000) The plasminogen activation system
  in tumor growth, invasion, and metastasis. Cell Mol Life Sci 57(1):25--40,
  \doi{10.1007/s000180050497}

\bibitem[{Angulo et~al(2012)Angulo, L{\'o}pez-Marcos, and Bees}]{Angulo2012}
Angulo O, L{\'o}pez-Marcos J, Bees M (2012) Mass structured systems with
  boundary delay: Oscillations and the effect of selective predation. Journal
  of Nonlinear Science 22(6):961--984, \doi{10.1007/s00332-012-9133-6}

\bibitem[{Armstrong et~al(2006)Armstrong, Painter, and
  Sherratt}]{Armstrong2006}
Armstrong NJ, Painter KJ, Sherratt JA (2006) A continuum approach to modelling
  cell--cell adhesion. J Theor Biol 243(1):98 -- 113,
  \doi{10.1016/j.jtbi.2006.05.030}

\bibitem[{Ayati(2000)}]{Ayati2000}
Ayati B (2000) A variable time step method for an age-dependent population
  model with nonlinear diffusion. SIAM Journal on Numerical Analysis
  37(5):1571--1589, \doi{10.1137/S003614299733010X}

\bibitem[{Ayati and Dupont(2002)}]{Ayati2002}
Ayati B, Dupont T (2002) Galerkin methods in age and space for a population
  model with nonlinear diffusion. SIAM Journal on Numerical Analysis
  40(3):1064--1076, \doi{10.1137/S0036142900379679}

\bibitem[{Ayati et~al(2006)Ayati, Webb, and Anderson}]{Ayati2006}
Ayati B, Webb G, Anderson A (2006) Computational methods and results for
  structured multiscale models of tumor invasion. Multiscale Modeling \&
  Simulation 5(1):1--20, \doi{10.1137/050629215}

\bibitem[{Ayati(2006)}]{Ayati2006a}
Ayati BP (2006) A structured-population model of proteus mirabilis swarm-colony
  development. J Math Biol 52(1):93--114, \doi{10.1007/s00285-005-0345-3}

\bibitem[{Basse and Ubezio(2007)}]{Basse2007}
Basse B, Ubezio P (2007) A generalised age- and phase-structured model of human
  tumour cell populations both unperturbed and exposed to a range of cancer
  therapies. Bull Math Biol 69(5):1673--1690, \doi{10.1007/s11538-006-9185-6}

\bibitem[{Basse et~al(2003)Basse, Baguley, Marshall, Joseph, van Brunt, Wake,
  and Wall}]{Basse2003}
Basse B, Baguley BC, Marshall ES, Joseph WR, van Brunt B, Wake G, Wall DJN
  (2003) A mathematical model for analysis of the cell cycle in cell lines
  derived from human tumors. J Math Biol 47(4):295--312,
  \doi{10.1007/s00285-003-0203-0}

\bibitem[{Basse et~al(2004)Basse, Baguley, Marshall, Joseph, van Brunt, Wake,
  and Wall}]{Basse2004}
Basse B, Baguley BC, Marshall ES, Joseph WR, van Brunt B, Wake G, Wall DJ
  (2004) Modelling cell death in human tumour cell lines exposed to the
  anticancer drug paclitaxel. J Math Biol 49(4):329--357,
  \doi{10.1007/s00285-003-0254-2}

\bibitem[{Basse et~al(2005)Basse, Baguley, Marshall, Wake, and
  Wall}]{Basse2005}
Basse B, Baguley B, Marshall E, Wake G, Wall D (2005) Modelling the flow of
  cytometric data obtained from unperturbed human tumour cell lines: Parameter
  fitting and comparison. Bull Math Biol 67(4):815--830,
  \doi{10.1016/j.bulm.2004.10.003}

\bibitem[{B{\'e}lair et~al(1995)B{\'e}lair, Mackey, and Mahaffy}]{Belair1995}
B{\'e}lair J, Mackey MC, Mahaffy JM (1995) Age-structured and two-delay models
  for erythropoiesis. Math Biosci 128(1--2):317 -- 346,
  \doi{10.1016/0025-5564(94)00078-E}

\bibitem[{Bernard et~al(2003)Bernard, Pujo-Menjouet, and Mackey}]{Bernard2003}
Bernard S, Pujo-Menjouet L, Mackey MC (2003) Analysis of cell kinetics using a
  cell division marker: Mathematical modeling of experimental data. Biophys J
  84(5):3414 -- 3424, \doi{10.1016/S0006-3495(03)70063-0}

\bibitem[{Billy et~al(2014)Billy, Clairambaultt, Fercoq, Gaubertt, Lepoutre,
  Ouillon, and Saito}]{Billy2014}
Billy F, Clairambaultt J, Fercoq O, Gaubertt S, Lepoutre T, Ouillon T, Saito S
  (2014) Synchronisation and control of proliferation in cycling cell
  population models with age structure. Mathematics and Computers in Simulation
  96:66 -- 94, \doi{10.1016/j.matcom.2012.03.005}

\bibitem[{Busenberg and Iannelli(1983)}]{Busenberg1983}
Busenberg S, Iannelli M (1983) A class of nonlinear diffusion problems in
  age-dependent population dynamics. Nonlinear Analysis: Theory, Methods \&
  Applications 7(5):501 -- 529, \doi{10.1016/0362-546X(83)90041-X}

\bibitem[{Byrne and Preziosi(2004)}]{Byrne2004}
Byrne HM, Preziosi L (2004) Modelling solid tumour growth using the theory of
  mixtures. Math Med Biol 20:341--366, \doi{10.1093/imammb/20.4.341}

\bibitem[{Calsina and Salda{\~n}a(1995)}]{Calsina1995}
Calsina {\`A}, Salda{\~n}a J (1995) A model of physiologically structured
  population dynamics with a nonlinear individual growth rate. J Math Biol
  33(4):335--364, \doi{10.1007/BF00176377}

\bibitem[{de~Camino-Beck and Lewis(2009)}]{Camino-Beck2009}
de~Camino-Beck T, Lewis M (2009) Invasion with stage-structured coupled map
  lattices: Application to the spread of scentless chamomile. Ecol Model
  220(23):3394 -- 3403, \doi{10.1016/j.ecolmodel.2009.09.003}

\bibitem[{Chaplain and Lolas(2005)}]{Chaplain2005}
Chaplain M, Lolas G (2005) Mathematical modelling of cancer cell invasion of
  tissue: the role of the urokinase plasminogen activation system. Mathematical
  Models and Methods in Applied Sciences 15(11):1685--1734,
  \doi{10.1142/S0218202505000947}

\bibitem[{Chaplain and Lolas(2006)}]{Chaplain2006}
Chaplain MAJ, Lolas G (2006) Mathematical modelling of cancer invasion of
  tissue: Dynamic heterogeneity. Netw Heterog Media 1(3):399--439,
  \doi{10.3934/nhm.2006.1.399}

\bibitem[{Chapman et~al(2007)Chapman, Plank, James, and Basse}]{Chapman2007}
Chapman SJ, Plank MJ, James A, Basse B (2007) A nonlinear model of age and
  size-structured populations with applications to cell cycles. The ANZIAM
  Journal 49:151--169, \doi{10.1017/S144618110001275X}

\bibitem[{Cubellis et~al(1990)Cubellis, Wun, and Blasi}]{Cubellis1990}
Cubellis MV, Wun TC, Blasi F (1990) Receptor-mediated internalization and
  degradation of urokinase is caused by its specific inhibitor {PAI-1}. EMBO J
  9(4):1079--1085

\bibitem[{Cushing(1998)}]{Cushing1998}
Cushing JM (1998) An Introduction to Structured Population Dynamics, CBMS-NSF
  Regional Conference Series in Applied Mathematics, vol~71. SIAM,
  \doi{10.1137/1.9781611970005.ch2}

\bibitem[{Cusulin et~al(2005)Cusulin, Iannelli, and Marinoschi}]{Cusulin2005}
Cusulin C, Iannelli M, Marinoschi G (2005) Age-structured diffusion in a
  multi-layer environment. Nonlinear Analysis: Real World Applications 6(1):207
  -- 223, \doi{10.1016/j.nonrwa.2004.08.006}

\bibitem[{Daukste et~al(2012)Daukste, Basse, Baguley, and Wall}]{Daukste2012}
Daukste L, Basse B, Baguley B, Wall D (2012) Mathematical determination of cell
  population doubling times for multiple cell lines. Bull Math Biol
  74(10):2510--2534, \doi{10.1007/s11538-012-9764-7}

\bibitem[{Deakin and Chaplain(2013)}]{Deakin2013}
Deakin N, Chaplain MAJ (2013) Mathematical modelling of cancer invasion: The
  role of membrane-bound matrix metalloproteinases. Frontiers in Oncology
  3(70), \doi{10.3389/fonc.2013.00070}

\bibitem[{Deisboeck et~al(2011)Deisboeck, Wang, Macklin, and
  Cristini}]{Deisboeck2011}
Deisboeck TS, Wang Z, Macklin P, Cristini V (2011) Multiscale cancer modeling.
  Annu Rev Biomed Eng 13:127--155, \doi{10.1146/annurev-bioeng-071910-124729}

\bibitem[{Delgado et~al(2006)Delgado, Molina-Becerra, and
  Su{\'a}rez}]{Delgado2006}
Delgado M, Molina-Becerra M, Su{\'a}rez A (2006) A nonlinear age-dependent
  model with spatial diffusion. Journal of Mathematical Analysis and
  Applications 313(1):366 -- 380, \doi{10.1016/j.jmaa.2005.09.042}

\bibitem[{Deng and Hallam(2006)}]{Deng2006}
Deng Q, Hallam TG (2006) An age structured population model in a spatially
  heterogeneous environment: Existence and uniqueness theory. Nonlinear
  Analysis: Theory, Methods \& Applications 65(2):379 -- 394,
  \doi{10.1016/j.na.2005.06.019}

\bibitem[{Di~Blasio(1979)}]{Di-Blasio1979}
Di~Blasio G (1979) Non-linear age-dependent population diffusion. J Math Biol
  8(3):265--284, \doi{10.1007/BF00276312}

\bibitem[{Diekmann and Metz(1994)}]{Diekmann1994}
Diekmann O, Metz J (1994) On the reciprocal relationship between life histories
  and population dynamics. In: Levin S (ed) Frontiers in Mathematical Biology,
  Lecture Notes in Biomathematics, vol 100, Springer Berlin Heidelberg, pp
  263--279, \doi{10.1007/978-3-642-50124-1_16}

\bibitem[{Diekmann et~al(1984)Diekmann, Heijmans, and Thieme}]{Diekmann1984}
Diekmann O, Heijmans H, Thieme H (1984) On the stability of the cell size
  distribution. J Math Biol 19(2):227--248, \doi{10.1007/BF00277748}

\bibitem[{Diekmann et~al(1992)Diekmann, Gyllenberg, Metz, and
  Thieme}]{Diekmann1992}
Diekmann O, Gyllenberg M, Metz JAJ, Thieme H (1992) The 'Cumulative'
  Formulation of (Physiologically) Structured Population Models. CWI

\bibitem[{Domschke et~al(2014)Domschke, Trucu, Gerisch, and
  Chaplain}]{Domschke2014}
Domschke P, Trucu D, Gerisch A, Chaplain MAJ (2014) Mathematical modelling of
  cancer invasion: Implications of cell adhesion variability for tumour
  infiltrative growth patterns. J Theor Biol 361:41--60,
  \doi{10.1016/j.jtbi.2014.07.010}

\bibitem[{Duffy(2004)}]{Duffy2004}
Duffy MJ (2004) The urokinase plasminogen activator system: Role in malignancy.
  Curr Pharm Des 10(1):39--49, \doi{10.2174/1381612043453559}

\bibitem[{Engwer et~al(2015)Engwer, Hillen, Knappitsch, and
  Surulescu}]{Engwer2015}
Engwer C, Hillen T, Knappitsch M, Surulescu C (2015) Glioma follow white matter
  tracts: a multiscale dti-based model. J Math Biol 71(3):551--582,
  \doi{10.1007/s00285-014-0822-7}

\bibitem[{Erban and Othmer(2005)}]{Erban2005}
Erban R, Othmer HG (2005) From signal transduction to spatial pattern formation
  in e. coli: A paradigm for multiscale modeling in biology. Multiscale
  Modeling {\&} Simulation 3(2):362--394, \doi{10.1137/040603565}

\bibitem[{Fitzgibbon et~al(1995)Fitzgibbon, Parrott, and Webb}]{Fitzgibbon1995}
Fitzgibbon W, Parrott M, Webb G (1995) Diffusion epidemic models with
  incubation and crisscross dynamics. Math Biosci 128(1--2):131 -- 155,
  \doi{10.1016/0025-5564(94)00070-G}

\bibitem[{von Foerster(1959)}]{Foerster1959}
von Foerster H (1959) Some remarks on changing populations. In: Stohlman JF
  (ed) The Kinetics of Cellular Proliferation, Grune and Stratton, New York, pp
  382--407

\bibitem[{Foley and Mackey(2009)}]{Foley2009}
Foley C, Mackey M (2009) Dynamic hematological disease: a review. J Math Biol
  58(1-2):285--322, \doi{10.1007/s00285-008-0165-3}

\bibitem[{F\"orste(1978)}]{Diekmann1982}
F\"orste J (1978) {Diekmann, O. / Temme, N. M. (Hrsg.), Nonlinear Diffusion
  Problems. Amsterdam. Mathematisch Centrum.} ZAMM 58(12):583--584,
  \doi{10.1002/zamm.19780581223}

\bibitem[{Gabriel et~al(2012)Gabriel, Garbett, Quaranta, Tyson, and
  Webb}]{Gabriel2012}
Gabriel P, Garbett SP, Quaranta V, Tyson DR, Webb GF (2012) The contribution of
  age structure to cell population responses to targeted therapeutics. J Theor
  Biol 311(0):19 -- 27, \doi{10.1016/j.jtbi.2012.07.001}

\bibitem[{Garroni and Langlais(1982)}]{Garroni1982}
Garroni MG, Langlais M (1982) Age-dependent population diffusion with external
  constraint. J Math Biol 14(1):77--94, \doi{10.1007/BF02154754}

\bibitem[{Gatenby and Gawlinski(1996)}]{Gatenby1996}
Gatenby RA, Gawlinski ET (1996) A reaction-diffusion model of cancer invasion.
  Cancer Res 56:5745--5753

\bibitem[{Gerisch and Chaplain(2008)}]{Gerisch2008}
Gerisch A, Chaplain M (2008) Mathematical modelling of cancer cell invasion of
  tissue: Local and non-local models and the effect of adhesion. J Theor Biol
  250(4):684 -- 704, \doi{10.1016/j.jtbi.2007.10.026}

\bibitem[{Gurtin and MacCamy(1981)}]{Gurtin1981}
Gurtin M, MacCamy R (1981) Diffusion models for age-structured populations.
  Math Biosci 54(1--2):49 -- 59, \doi{10.1016/0025-5564(81)90075-4}

\bibitem[{Gwiazda and Marciniak-Czochra(2010)}]{Gwiazda2010}
Gwiazda P, Marciniak-Czochra A (2010) Structured population equations in metric
  spaces. Journal of Hyperbolic Differential Equations 07(04):733--773,
  \doi{10.1142/S021989161000227X}

\bibitem[{Gyllenberg(1982)}]{Gyllenberg1982}
Gyllenberg M (1982) Nonlinear age-dependent population dynamics in continuously
  propagated bacterial cultures. Math Biosci 62(1):45 -- 74,
  \doi{10.1016/0025-5564(82)90062-1}

\bibitem[{Gyllenberg(1983)}]{Gyllenberg1983}
Gyllenberg M (1983) Stability of a nonlinear age-dependent population model
  containing a control variable. SIAM Journal on Applied Mathematics
  43(6):1418--1438, \urlprefix\url{http://www.jstor.org/stable/2101185}

\bibitem[{Gyllenberg(1986)}]{Gyllenberg1986}
Gyllenberg M (1986) The size and scar distributions of the yeast saccharomyces
  cerevisiae. J Math Biol 24(1):81--101, \doi{10.1007/BF00275722}

\bibitem[{Gyllenberg and Hanski(1997)}]{Gyllenberg1997a}
Gyllenberg M, Hanski I (1997) Habitat deterioration, habitat destruction, and
  metapopulation persistence in a heterogenous landscape. Theor Popul Biol
  52(3):198 -- 215, \doi{10.1006/tpbi.1997.1333}

\bibitem[{Gyllenberg and Webb(1987)}]{Gyllenberg1987}
Gyllenberg M, Webb G (1987) Age-size structure in populations with quiescence.
  Math Biosci 86(1):67 -- 95, \doi{10.1016/0025-5564(87)90064-2}

\bibitem[{Gyllenberg and Webb(1990)}]{Gyllenberg1990}
Gyllenberg M, Webb G (1990) A nonlinear structured population model of tumor
  growth with quiescence. J Math Biol 28(6):671--694, \doi{10.1007/BF00160231}

\bibitem[{Gyllenberg et~al(1997)Gyllenberg, Hanski, and
  Lindstr{\"o}m}]{Gyllenberg1997}
Gyllenberg M, Hanski I, Lindstr{\"o}m T (1997) Continuous versus discrete
  single species population models with adjustable reproductive strategies.
  Bull Math Biol 59(4):679--705, \doi{10.1007/BF02458425}

\bibitem[{Gyllenberg et~al(2002)Gyllenberg, Osipov, and
  P{\"a}iv{\"a}rinta}]{Gyllenberg2002}
Gyllenberg M, Osipov A, P{\"a}iv{\"a}rinta L (2002) The inverse problem of
  linear age-structured population dynamics. Journal of Evolution Equations
  2(2):223--239, \doi{10.1007/s00028-002-8087-9}

\bibitem[{Halmos(1978)}]{Halmos1978}
Halmos PR (1978) Measure Theory, 2nd edn. Springer

\bibitem[{Huang(1994)}]{Huang1994}
Huang C (1994) An age-dependent population model with nonlinear diffusion in
  {$\mathbf{R}^n$}. Quart Appl Math 52:377--398

\bibitem[{Huyer(1994)}]{Huyer1994}
Huyer W (1994) A size-structured population-model with dispersion. Journal of
  Mathematical Analysis and Applications 181(3):716 -- 754,
  \doi{10.1006/jmaa.1994.1054}

\bibitem[{Kelkel and Surulescu(2012)}]{Kelkel2012}
Kelkel J, Surulescu C (2012) A multiscale approach to cell migration in tissue
  networks. Mathematical Models and Methods in Applied Sciences
  22(03):1150,017, \doi{10.1142/S0218202511500175}

\bibitem[{Kunisch et~al(1985)Kunisch, Schappacher, and Webb}]{Kunisch1985}
Kunisch K, Schappacher W, Webb GF (1985) Nonlinear age-dependent population
  dynamics with random diffusion. Computers \& Mathematics with Applications
  11(1--3):155 -- 173, \doi{10.1016/0898-1221(85)90144-0}

\bibitem[{Langlais(1988)}]{Langlais1988}
Langlais M (1988) Large time behavior in a nonlinear age-dependent population
  dynamics problem with spatial diffusion. J Math Biol 26(3):319--346,
  \doi{10.1007/BF00277394}

\bibitem[{Langlais and Milner(2003)}]{Langlais2003}
Langlais M, Milner FA (2003) Existence and uniqueness of solutions for a
  diffusion model of host--parasite dynamics. Journal of Mathematical Analysis
  and Applications 279(2):463 -- 474, \doi{10.1016/S0022-247X(03)00020-9}

\bibitem[{Lauren{\c c}ot and Walker(2008)}]{Laurencot2008}
Lauren{\c c}ot P, Walker C (2008) An age and spatially structured population
  model for proteus mirabilis swarm-colony development. Math Model Nat Phenom
  3(7):49--77, \doi{10.1051/mmnp:2008041}

\bibitem[{Lewis et~al(2010)Lewis, Nelson, and Xu}]{Lewis2010}
Lewis M, Nelson W, Xu C (2010) A structured threshold model for mountain pine
  beetle outbreak. Bull Math Biol 72(3):565--589,
  \doi{10.1007/s11538-009-9461-3}

\bibitem[{Lodish et~al(2007)Lodish, Berk, Kaiser, Krieger, Scott, Bretscher,
  Ploegh, and Matsudaira}]{Lodish2007}
Lodish H, Berk A, Kaiser CA, Krieger M, Scott MP, Bretscher A, Ploegh H,
  Matsudaira P (2007) Molecular Cell Biology, 6th edn. W.H.Freeman

\bibitem[{MacCamy(1981)}]{MacCamy1981}
MacCamy R (1981) A population model with nonlinear diffusion. Journal of
  Differential Equations 39(1):52 -- 72, \doi{10.1016/0022-0396(81)90083-8}

\bibitem[{Mackey and Glass(1977)}]{Mackey1977}
Mackey M, Glass L (1977) Oscillation and chaos in physiological control
  systems. Science 197(4300):287--289, \doi{10.1126/science.267326}

\bibitem[{Macklin et~al(2009)Macklin, McDougall, Anderson, Chaplain, Cristini,
  and Lowengrub}]{Macklin2009}
Macklin P, McDougall SR, Anderson ARA, Chaplain MAJ, Cristini V, Lowengrub J
  (2009) Multiscale modelling and nonlinear simulation of vascular tumour
  growth. J Math Biol 58:765--798, \doi{10.1007/s00285-008-0216-9}

\bibitem[{Mahaffy et~al(1998)Mahaffy, B{\'e}lair, and Mackey}]{Mahaffy1998}
Mahaffy JM, B{\'e}lair J, Mackey MC (1998) Hematopoietic model with moving
  boundary condition and state dependent delay: Applications in erythropoiesis.
  J Theor Biol 190(2):135 -- 146, \doi{10.1006/jtbi.1997.0537}

\bibitem[{Marciniak-Czochra and Ptashnyk(2008)}]{Marciniak-Czochra2008}
Marciniak-Czochra A, Ptashnyk M (2008) Derivation of a macroscopic
  receptor-based model using homogenization techniques. SIAM Journal on
  Mathematical Analysis 40(1):215--237, \doi{10.1137/050645269}

\bibitem[{Matter et~al(2002)Matter, Hanski, and Gyllenberg}]{Matter2002}
Matter SF, Hanski I, Gyllenberg M (2002) A test of the metapopulation model of
  the species--area relationship. Journal of Biogeography 29(8):977--983,
  \doi{10.1046/j.1365-2699.2002.00748.x}

\bibitem[{Mercker et~al(2013)Mercker, Marciniak-Czochra, Richter, and
  Hartmann}]{Mercker2013}
Mercker M, Marciniak-Czochra A, Richter T, Hartmann D (2013) Modeling and
  computing of deformation dynamics of inhomogeneous biological surfaces. SIAM
  Journal on Applied Mathematics 73(5):1768--1792, \doi{10.1137/120885553}

\bibitem[{Metz and Diekmann(1986)}]{Metz1986}
Metz JAJ, Diekmann O (1986) The Dynamics of Physiologically Structured
  Populations, Lecture Notes in Biomathematics, vol~68. Springer-Verlag

\bibitem[{Othmer and Xue(2013)}]{Othmer2013}
Othmer HG, Xue C (2013) Dispersal, Individual Movement and Spatial Ecology: A
  Mathematical Perspective, Springer Berlin Heidelberg, Berlin, Heidelberg,
  chap The Mathematical Analysis of Biological Aggregation and Dispersal:
  Progress, Problems and Perspectives, pp 79--127.
  \doi{10.1007/978-3-642-35497-7_4}

\bibitem[{Othmer et~al(1988)Othmer, Dunbar, and Alt}]{Othmer1988}
Othmer HG, Dunbar SR, Alt W (1988) Models of dispersal in biological systems. J
  Math Biol 26(3):263--298, \doi{10.1007/BF00277392}

\bibitem[{Parsons et~al(1997)Parsons, Watson, Brown, Collins, and
  Steele}]{Parsons1997}
Parsons SL, Watson SA, Brown PD, Collins HM, Steele RJ (1997) Matrix
  metalloproteinases. Brit J Surg 84(2):160--166,
  \doi{10.1046/j.1365-2168.1997.02719.x}

\bibitem[{Pepper(2001)}]{Pepper2001}
Pepper MS (2001) Role of the matrix metalloproteinase and plasminogen
  activator-plasmin systems in angiogenesis. Arterioscl Throm Vas
  21(7):1104--1117, \doi{10.1161/hq0701.093685}

\bibitem[{Perthame(2007)}]{Perthame2007}
Perthame B (2007) Transport Equations in Biology. Frontiers in Mathematics,
  {Birkh\"auser}

\bibitem[{Ramis-Conde et~al(2008)Ramis-Conde, Drasdo, Anderson, and
  Chaplain}]{Ramis-Conde2008}
Ramis-Conde I, Drasdo D, Anderson ARA, Chaplain MAJ (2008) Modeling the
  influence of the e-cadherin-{$\beta$}-catenin pathway in cancer cell
  invasion: A multiscale approach. Biophys J 95(1):155--165,
  \doi{10.1529/biophysj.107.114678}

\bibitem[{Rhandi(1998)}]{Rhandi1998}
Rhandi A (1998) Positivity and stability for a population equation with
  diffusion on $l^1$. Positivity 2(2):101--113, \doi{10.1023/A:1009721915101}

\bibitem[{Rhandi and Schnaubelt(1999)}]{Rhandi1999}
Rhandi A, Schnaubelt R (1999) Asymptotic behaviour of a non-autonomous
  population equation with diffusion in {$L^1$}. Discrete and Continuous
  Dynamical Systems - Series A 5(3):663--683, \doi{10.3934/dcds.1999.5.663}

\bibitem[{Roeder et~al(2009)Roeder, Herberg, and Horn}]{Roeder2009}
Roeder I, Herberg M, Horn M (2009) An ``age'' structured model of hematopoietic
  stem cell organization with application to chronic myeloid leukemia. Bull
  Math Biol 71(3):602--626, \doi{10.1007/s11538-008-9373-7}

\bibitem[{de~Roos(1997)}]{Roos1997}
de~Roos AM (1997) A gentle introduction to physiologically structured
  population models. In: Tuljapurkar S, Caswell H (eds) Structured-Population
  Models in Marine, Terrestrial, and Freshwater Systems, Population and
  Community Biology Series, vol~18, Springer US, pp 119--204,
  \doi{10.1007/978-1-4615-5973-3_5}

\bibitem[{Sinko and Streifer(1967)}]{Sinko1967}
Sinko JW, Streifer W (1967) A new model for age-size structure of a population.
  Ecology 48(6):910--918, \doi{10.2307/1934533}

\bibitem[{Skellam(1951)}]{Skellam1951}
Skellam JG (1951) Random dispersal in theoretical populations. Biometrika
  38(1/2):196--218, \doi{10.2307/2332328}

\bibitem[{So et~al(2001)So, Wu, and Zou}]{So2001}
So JWH, Wu J, Zou X (2001) A reaction-diffusion model for a single species with
  age structure. i travelling wavefronts on unbounded domains. Proceedings of
  the Royal Society of London A: Mathematical, Physical and Engineering
  Sciences 457(2012):1841--1853, \doi{10.1098/rspa.2001.0789}

\bibitem[{Trucco(1965{\natexlab{a}})}]{Trucco1965a}
Trucco E (1965{\natexlab{a}}) Mathematical models for cellular systems the von
  foerster equation. part i. The Bulletin of Mathematical Biophysics
  27(3):285--304, \doi{10.1007/BF02478406}

\bibitem[{Trucco(1965{\natexlab{b}})}]{Trucco1965b}
Trucco E (1965{\natexlab{b}}) Mathematical models for cellular systems. the von
  foerster equation. part ii. The Bulletin of Mathematical Biophysics
  27(4):449--471, \doi{10.1007/BF02476849}

\bibitem[{Trucu et~al(2013)Trucu, Lin, Chaplain, and Wang}]{Trucu2013}
Trucu D, Lin P, Chaplain MAJ, Wang Y (2013) A multiscale moving boundary model
  arising in cancer invasion. Multiscale Model Sim 11(1):309--335,
  \doi{10.1137/110839011}

\bibitem[{Tucker and Zimmerman(1988)}]{Tucker1988}
Tucker SL, Zimmerman SO (1988) A nonlinear model of population dynamics
  containing an arbitrary number of continuous structure variables. SIAM
  Journal on Applied Mathematics 48(3):pp. 549--591,
  \urlprefix\url{http://www.jstor.org/stable/2101595}

\bibitem[{Ulisse et~al(2009)Ulisse, Baldini, Sorrenti, and
  D'Armiento}]{Ulisse2009}
Ulisse S, Baldini E, Sorrenti S, D'Armiento M (2009) The urokinase plasminogen
  activator system: A target for anti-cancer therapy. Curr Cancer Drug Targets
  9(1):32--71, \doi{10.2174/156800909787314002}

\bibitem[{Walker(2007)}]{Walker2007}
Walker C (2007) Global well-posedness of a haptotaxis model with spatial and
  age structure. Differential and Integral Equations 20(9):1053--1074,
  \urlprefix\url{http://projecteuclid.org/euclid.die/1356039311}

\bibitem[{Walker(2008)}]{Walker2008}
Walker C (2008) Global existence for an age and spatially structured haptotaxis
  model with nonlinear age-boundary conditions. European Journal of Applied
  Mathematics 19:113--147, \doi{10.1017/S095679250800733X}

\bibitem[{Walker(2009)}]{Walker2009}
Walker C (2009) Positive equilibrium solutions for age- and
  spatially-structured population models. SIAM Journal on Mathematical Analysis
  41(4):1366--1387, \doi{10.1137/090750044}

\bibitem[{Webb(2008)}]{Webb2008}
Webb G (2008) Population models structured by age, size, and spatial position.
  In: Magal P, Ruan S (eds) Structured Population Models in Biology and
  Epidemiology, Lecture Notes in Mathematics, vol 1936, Springer Berlin
  Heidelberg, pp 1--49, \doi{10.1007/978-3-540-78273-5_1}

\bibitem[{Webb(1985)}]{Webb1985}
Webb GF (1985) Theory of Nonlinear Age-dependent Population Dynamics, Pure and
  Applied Mathematics, vol~89. Marcel Dekker, New York

\bibitem[{Xue(2015)}]{Xue2015}
Xue C (2015) Macroscopic equations for bacterial chemotaxis: integration of
  detailed biochemistry of cell signaling. J Math Biol 70(1):1--44,
  \doi{10.1007/s00285-013-0748-5}

\bibitem[{Xue et~al(2009)Xue, Othmer, and Erban}]{Xue2009}
Xue C, Othmer HG, Erban R (2009) From individual to collective behavior of
  unicellular organisms: Recent results and open problems. AIP Conference
  Proceedings 1167(1):3--14, \doi{10.1063/1.3246413}

\bibitem[{Xue et~al(2011)Xue, Hwang, Painter, and Erban}]{Xue2011}
Xue C, Hwang HJ, Painter KJ, Erban R (2011) Travelling waves in hyperbolic
  chemotaxis equations. Bull Math Biol 73(8):1695--1733,
  \doi{10.1007/s11538-010-9586-4}

\bibitem[{Yang et~al(2006)Yang, Avila, Wang, Trevino, Gallick, Kitadai, Sasaki,
  and Boyd}]{Yang2006}
Yang L, Avila H, Wang H, Trevino J, Gallick GE, Kitadai Y, Sasaki T, Boyd DD
  (2006) Plasticity in urokinase-type plasminogen activator receptor (upar)
  display in colon cancer yields metastable subpopulations oscillating in cell
  surface upar density -- implications in tumor progression. Cancer Research
  66(16):7957--7967, \doi{10.1158/0008-5472.CAN-05-3208}

\end{thebibliography}

\end{document}